\pdfoutput=1 

\documentclass[11pt]{article}

\usepackage[utf8]{inputenc} 
\usepackage[T1]{fontenc}    
\usepackage{url}            
\usepackage{booktabs}       
\usepackage{amsfonts}       
\usepackage{nicefrac}       
\usepackage{microtype}      
\usepackage{amsmath, amssymb, amsthm}
\usepackage{mathtools}      
\usepackage{float}
\usepackage{enumitem}
\usepackage{xspace}

\usepackage{fullpage}
\usepackage[dvipsnames]{xcolor}
\usepackage{hyperref}
\usepackage{pdfsync}
\usepackage{bm}
\usepackage[mathscr]{eucal} 
\usepackage{comment}
\usepackage{thm-restate}
\usepackage[capitalize, nameinlink]{cleveref}
\usepackage{float}
\usepackage{algorithm}
\usepackage[noend]{algpseudocode}
\usepackage{parskip}

\makeatletter
\def\thm@space@setup{%
  \thm@preskip=\parskip \thm@postskip=0pt
}
\makeatother

\newcommand{\mat}[1]{\mathbf{#1}}

\newcommand{\tensor}[1]{\bm{\mathscr{#1}}}
\newcommand{\vectorize}{\textnormal{vec}}
\newcommand{\frobenius}{\textnormal{F}}

\newcommand{\var}{\textnormal{Var}}
\newcommand{\ApproxRidgeRegression}{\textsc{ApproximateRidgeRegression}}

\newcommand{\ktimes}{\otimes}

\newcommand{\opt}{\textnormal{opt}}
\newcommand{\rank}{\textnormal{rank}}
\newcommand{\eff}{\textnormal{eff}}

\newcommand{\R}{\mathbb{R}}

\newcommand{\E}{\mathbb{E}}

\DeclareMathOperator*{\argmin}{arg\,min}
\DeclareMathOperator*{\defeq}{\overset{def}{=}}

\newcommand{\declarecolor}[2]{\definecolor{#1}{RGB}{#2}\expandafter\newcommand\csname #1\endcsname[1]{\textcolor{#1}{##1}}}
\declarecolor{White}{255, 255, 255}
\declarecolor{Black}{0, 0, 0}
\declarecolor{LightGray}{216, 216, 216}
\declarecolor{Gray}{127, 127, 127}
\declarecolor{Orange}{237, 125, 49}
\declarecolor{LightOrange}{251,229, 214}
\declarecolor{Yellow}{255, 192, 0}
\declarecolor{LightYellow}{255, 242, 200}
\declarecolor{Blue}{91, 155, 213}
\declarecolor{LightBlue}{222, 235, 247}
\declarecolor{Green}{112, 173, 71}
\declarecolor{LightGreen}{226, 240, 217}
\declarecolor{Navy}{68, 114, 196}
\declarecolor{LightNavy}{218, 227, 243}

\hypersetup{
	colorlinks=true,
	pdfpagemode=UseNone,
	citecolor=Navy,
	linkcolor=Navy,
	urlcolor=Navy,
}

\DeclarePairedDelimiter{\set}{\{}{\}}
\DeclarePairedDelimiter{\parens}{(}{)}
\DeclarePairedDelimiter{\bracks}{[}{]}

\DeclarePairedDelimiter{\ceil}{\lceil}{\rceil}
\DeclarePairedDelimiter{\norm}{\lVert}{\rVert}

 \newcommand{\cD}{\mathcal{D}}

 \newcommand{\cR}{\mathcal{R}}

\theoremstyle{plain}
\newtheorem{theorem}{Theorem}[section]
\newtheorem{lemma}[theorem]{Lemma}

\theoremstyle{definition}
\newtheorem{definition}[theorem]{Definition}

\crefformat{equation}{(#2#1#3)}

\usepackage[round]{natbib}
\title{Fast Low-Rank Tensor Decomposition by Ridge Leverage Score Sampling}

\author{%
  Matthew Fahrbach\\
  Google Research\\
  \texttt{\href{mailto:fahrbach@google.com}{fahrbach@google.com}}
  \and
  Mehrdad Ghadiri\\
  Georgia Tech\\
  \texttt{\href{mailto:ghadiri@gatech.edu}{ghadiri@gatech.edu}}
  \and
  Thomas Fu\\
  Google Research\\
  \texttt{\href{mailto:thomasfu@google.com}{thomasfu@google.com}}
}

\begin{document}

\maketitle

\begin{abstract}
Low-rank tensor decomposition generalizes low-rank matrix approximation
and
is a powerful technique for discovering
low-dimensional structure in high-dimensional data.
In this paper, we study Tucker decompositions
and use tools from randomized numerical linear algebra
called \emph{ridge leverage scores} to accelerate the core tensor update
step in the widely-used alternating least squares (ALS) algorithm.
Updating the core tensor, a severe bottleneck in ALS,
is a highly-structured ridge regression problem
where the design matrix is a Kronecker product of the factor matrices.
We show how to use approximate ridge leverage scores to construct a
sketched instance for any ridge regression problem such that the solution
vector for the sketched problem is a $(1+\varepsilon)$-approximation to the
original instance.
Moreover, we show that classical leverage scores suffice as an approximation,
which then allows us to exploit the Kronecker structure and
update the core tensor in time
that depends predominantly on the rank and the sketching
parameters (i.e., sublinear in the size of the input tensor).
We also give upper bounds for ridge leverage scores as rows are removed
from the design matrix (e.g., if the tensor has missing entries),
and we demonstrate the effectiveness of our approximate ridge regression
algorithm for large, low-rank Tucker decompositions on both synthetic and
real-world data.
\end{abstract}

\newpage

\section{Introduction}
\label{sec:introduction}

Tensor decomposition has a rich multidisciplinary history,
but it has only recently become ubiquitous
due to a surge of applications in data mining, machine learning, and
signal processing~\cite{kolda2009tensor,rabanser2017introduction,sidiropoulos2017tensor}.
The most prevalent tensor decompositions are the CP and Tucker decompositions,
which are often thought of as generalized singular value decompositions.
Consequently, there are natural notions of low-rank tensor decomposition.
Unlike matrix factorization, however, even computing the analogs of rank
for a tensor is NP-hard~\cite{hillar2013most}.
Therefore, most low-rank tensor decomposition algorithms fix the rank
structure in advance and then optimize the variables of the decomposition
to fit the data.
While conceptually simple, this
technique can be extremely effective for many real-world applications
since high-dimensional data is often inherently low-dimensional.

One of the cornerstones of low-rank tensor decomposition is
the alternating least squares (ALS) algorithm.
For CP and Tucker decompositions, ALS cyclically
optimizes disjoint blocks of variables
while keeping all others fixed.
If no additional constraints or regularization are imposed,
then each step of ALS is an ordinary least squares problem.
For Tucker decompositions, 
the higher order singular value decomposition
(HOSVD) and higher order orthogonal iteration (HOOI) algorithms~\cite{kolda2009tensor}
are popular alternatives,
but (1) do not scale as easily since they
compute SVDs of matricizations of the data tensor in each step,
and (2) do work well for data with missing entries.
Therefore, we focus on ALS for Tucker decompositions
and incorporate a new technique from randomized numerical linear algebra 
called \emph{ridge leverage score sampling} to speed up its bottleneck step.
Recently, \citet{cheng2016spals} and \citet{larsen2020practical} showed that
(statistical) leverage score sampling is useful for accelerating ALS
for CP decompositions.
Leverage scores measure the importance of each observation in a least squares
problem,
and were recently generalized by \citet{ahmed2015fast} to account for
Tikhonov regularization, hence the name ridge leverage scores.
The CP decomposition algorithms in \cite{cheng2016spals,larsen2020practical},
as well as the work of \citet{huaian2019optimal} on
Kronecker product regression, exploit
the fact that the necessary leverage score distribution can be
closely approximated by a related product distribution.
This ultimately leads to efficient sampling subroutines.
We follow a similar approach to give a fast sampling-based ALS algorithm for
Tucker decompositions, and we show how to seamlessly extend
leverage score sampling methods
to account for L2 regularization (i.e., approximate ridge regression).

\subsection{Our Contributions}

This work gives several new results for ridge leverage score sampling
and its applications in low-rank tensor decomposition.
Below is a summary of our contributions:

\begin{enumerate}
    \item Our first result is a method for augmenting
      approximate ridge leverage score distributions such that
      (1) we can sample from the augmented distribution in the
      same amount of time as the original distribution, and
      (2) if we sample rows according to this augmented distribution, then we
      can construct a sketched version for any ridge regression problem
      such that the solution vector for the sketch
      gives a $(1+\varepsilon)$-approximation to
      the original instance.
      We then show that the statistical leverage scores of the design matrix
      for any ridge regression problem are a useful overestimate of the
      $\lambda$-ridge leverage scores when augmented.
      Moreover, we quantify how the sample complexity of sketching
      algorithms decreases as the value of $\lambda$ increases
      (i.e., as the \emph{effective dimensionality} of the ridge regression
      problem shrinks).

    \item Our second key result explores how
      ridge leverage score sampling can be used to compute low-rank
      Tucker decompositions of tensors.
      We consider the ubiquitous alternating least squares algorithm
      and speed up its core tensor update step---a notorious
      bottleneck for Tucker decomposition algorithms.
      We use our approximate ridge regression subroutine and exploit the
      fact that the design matrix in every core tensor update is a Kronecker
      product of the factor matrices.
      In particular, this means that the leverage score distribution of the
      design matrix is a product distribution of the leverage score
      distributions for the factor matrices, hence we can sample rows
      from the augmented distribution
      in time sublinear in the number of its rows.
      Our core tensor update is designed to be fast both in theory
      and in practice, since its time complexity is predominantly a function of
      the rank and sketching parameters.

    \item Next, as a step towards better understanding alternating least
      squares for tensor completion, we derive upper bounds
      for the $\lambda$-ridge leverage scores when rows are removed from the
      design matrix (i.e., if the tensor has missing entries).
      While these bounds hold for general matrices and can be pessimistic
      if a large fraction of rows are removed,
      the proofs provide useful insight into exactly how $\lambda$-ridge
      leverage scores generalize classical leverage scores.
      
    \item Lastly, we demonstrate how our approximate core tensor update
      based on fast (factored) leverage score sampling leads
      to massive improvements in the
      running time for low-rank Tucker decompositions while preserving the
      original solution quality of ALS.
      Specifically,
      we profile this algorithm using large, dense synthetic tensors
      and the movie data of~\citet{malik2018low},
      which explores sketching Tucker decompositions in a data stream model.
\end{enumerate}

\subsection{Related Works}

\paragraph{Tensor Decomposition.}
The algorithms of~\citet{cheng2016spals} and \citet{larsen2020practical}
that use leverage score sampling with ALS to compute (unregularized) CP
decompositions are most directly related.
Avoiding degeneracies in a CP decomposition using ALS
has carefully been studied in~\cite{comon2009tensor}.
For first order methods, a step of gradient descent typically takes as
long as an iteration of ALS since both involve computing the same
quantities~\cite[Remark 3]{sidiropoulos2017tensor}.
SGD-based methods, however, are known to be efficient for certain structured CP
decompositions~\cite{ge2015escaping}.
For Tucker decompositions, \citet{frandsen2020optimization} recently showed
that if the tensor being learned has an exact Tucker decomposition, then all
local minima are globally optimal.
Streaming algorithms for Tucker decompositions based on sketching and
SGD have also recently been explored
by~\citet{malik2018low} and \citet{traore2019singleshot}, respectively.
The more general problem of low-rank tensor completion
is a fundamental approach for estimating the values of
missing data~\cite{acar2011scalable,jain2013low,jain2014provable,filipovic2015tucker},
and has led to notable breakthroughs in computer vision~\cite{liu2012tensor}.
Popular approaches for
tensor completion are based Riemannian optimization~\cite{kressner2014low,kasai2016low,madhav2018dual}
alternating least squares~\cite{zhou2013tensor,grasedyck2015variants,liu2020tensor},
and projected gradient methods~\cite{yu2016learning}.

\paragraph{Ridge Leverage Scores.}
\citet{ahmed2015fast} recently extended the idea of statistical leverage
scores to the setting of ridge regression,
and used these scores to derive a sampling distribution that reduces the sketch size
(i.e., the number of sampled rows) to the effective dimension of the problem.
Since then, sampling from approximate ridge leverage score distributions
has played a critical role in fundamental works for
sparse low-rank matrix approximation~\cite{cohen2017input},
an improved Nystr\"om method via recursive sampling~\cite{musco2017recursive},
bounding the statistical risk of ridge regression~\cite{mccurdy2018ridge},
a new sketching-based iterative method for ridge regression~\cite{chowdhury2018iterative},
and improved bounds for the number of random Fourier features needed for
ridge regression as a function of the effective dimension~\cite{li2019towards}.
Closely related are fast recursive algorithms for
computing approximate leverage scores~\cite{cohen2015uniform} and
for solving overconstrained least squares~\cite{li2013iterative}.
Recent works have also explored sketching for sparse Kronecker product
regression, which exploit a similar product distribution property of 
leverage scores~\cite{diao2018sketching,huaian2019optimal}.

\section{Preliminaries}
\label{sec:preliminaries}


\paragraph{Notation and Terminology.}
The \emph{order} of a tensor is the number of its dimensions,
also known as \emph{ways} or \emph{modes}.
Scalars (zeroth-order tensors) are denoted by normal lowercase letters $x \in \R$,
vectors (first-order tensors) by boldface lowercase letters $\mat{x} \in \R^{n}$,
and matrices (second-order tensors) by boldface uppercase letters $\mat{X} \in \R^{m \times n}$.
Higher-order tensors are denoted by boldface script letters
$\tensor{X} \in \R^{I_1 \times I_2 \times \cdots \times I_N}$.
For higher-order tensors,
we use normal uppercase letters to denote
the size of an index set (e.g., $[N] = \{1,2,\dots,N\}$).
The $i$-th entry of a vector $\mat{x}$ is denoted by $x_i$,
the $(i,j)$-th entry of a matrix $\mat{X}$ by $x_{ij}$,
and the $(i,j,k)$-th entry of a third-order tensor $\tensor{X}$ by $x_{ijk}$.

\paragraph{Linear Algebra.}
Let $\mat{I}_{n}$ denote the $n \times n$ identity matrix and
$\mat{0}_{m \times n}$ denote the $m \times n$ zero matrix.
Denote the transpose of a matrix $\mat{A} \in \R^{m \times n}$
by $\mat{A}^\intercal$ and the Moore--Penrose inverse by $\mat{A}^+$.
The singular value decomposition (SVD) of $\mat{A}$ is a factorization
of the form $\mat{U} \mat{\Sigma} \mat{V}^\intercal$,
where $\mat{U} \in \R^{m \times m}$ and $\mat{V} \in \R^{n \times n}$
are orthogonal matrices,
and $\mat{\Sigma} \in \R^{m \times n}$ is a diagonal matrix with
non-negative real numbers on its diagonal.
The entries $\sigma_{i}(\mat{A})$ of $\mat{\Sigma}$
are the singular values of $\mat{A}$,
and the number of non-zero singular values is equal to $r = \rank(\mat{A})$.
The \emph{compact SVD} is a similar decomposition 
where $\mat{\Sigma} \in \R^{r \times r}$ is a
diagonal matrix containing only the non-zero singular values.
Lastly, we denote the Kronecker product of two matrices $\mat{A} \in \R^{m \times n}$
and $\mat{B} \in \R^{p \times q}$ by $\mat{A} \ktimes \mat{B} \in \R^{(mp)\times(nq)}$.

\paragraph{Tensor Products.}
The \emph{fibers} of a tensor are vectors created by fixing all but one
index
(e.g., for a third-order tensor $\tensor{X}$, the column, row, and tube fibers are
denoted by $\mat{x}_{:jk}$, $\mat{x}_{i:k}$, and $\mat{x}_{ij:}$,
respectively).
The \emph{mode-$n$ unfolding} of a tensor
$\tensor{X} \in \R^{I_1 \times I_2 \times \cdots \times I_N}$
is the matrix $\mat{X}_{(n)} \in \R^{I_n \times (I_1 \cdots I_{n-1}I_{n+1}\cdots I_N)}$
that arranges the mode-$n$ fibers of $\tensor{X}$ as the columns of
$\mat{X}_{(n)}$ ordered lexicographically by index.
Going one step further,
the \emph{vectorization} of $\tensor{X} \in \R^{I_1 \times I_2 \times \cdots \times I_N}$
is the vector $\vectorize(\tensor{X}) \in \R^{I_1 I_2 \cdots I_N}$
formed by vertically stacking the entries of $\tensor{X}$
ordered lexicographically by index
(e.g., this transforms matrix $\mat{X} \in \R^{m \times n}$ into a
tall vector $\vectorize(\mat{X})$ by stacking its columns).

The \emph{$n$-mode product} of a tensor $\tensor{X} \in \R^{I_1\times I_2 \times \cdots \times I_N}$
and matrix $\mat{A} \in \R^{J \times I_n}$ is denoted by
$\tensor{Y} = \tensor{X} \times_{n} \mat{A}$ with
$\tensor{Y} \in \R^{I_1\times \cdots \times I_{n-1} \times J \times I_{n+1} \times \cdots \times I_N}$.
Intuitively,
this operation multiplies each mode-$n$ fiber of
$\tensor{X}$ by the matrix $\mat{A}$.
Elementwise, this operation is expressed as follows:
\[
    \parens*{\tensor{X} \times_{n} \mat{A}}_{i_1\dots i_{n-1} j i_{n+1} \dots i_{N}}
    =
    \sum_{i_n=1}^{I_n} x_{i_1 i_2 \dots i_N} a_{j i_n}.
\]
The Frobenius norm $\norm{\tensor{X}}_{\frobenius}$
of a tensor $\tensor{X}$ is the square root of the sum of
the squares of all its entries.

\paragraph{Tucker Decomposition.}
The \emph{Tucker decomposition} decomposes a tensor
$\tensor{X} \in \R^{I_1 \times I_2 \times \cdots \times I_N}$
into a \emph{core tensor}
$\tensor{G} \in \R^{R_1 \times R_2 \times \cdots \times R_N}$
and multiple \emph{factor matrices}
$\mat{A}^{(n)} \in \R^{I_n \times R_n}$.
We can express the problem of finding a
Tucker decomposition of $\tensor{X}$ as minimizing the loss function
\begin{align*}
  L\parens*{\tensor{G}, \mat{A}^{(1)}, \dots, \mat{A}^{(N)}}
  =
  \norm*{\tensor{X} - \tensor{G} \times_{1}\mat{A}^{(1)}
     \times_{2} \cdots
     \times_{N} \mat{A}^{(N)}
  }_{\frobenius}^2
  + \lambda \parens*{
    \norm*{\tensor{G}}_{\frobenius}^2
    +
    \sum_{n=1}^N
      \norm*{\mat{A}^{(n)}}_{\frobenius}^2
  },
\end{align*}
where $\lambda$ is a regularization parameter.
The elements of 
$\widehat{\tensor{X}} \defeq \tensor{G} \times_{1}\mat{A}^{(1)}
     \times_{2} \mat{A}^{(2)} \times_{3} \cdots
     \times_{N} \mat{A}^{(N)}$
are
\begin{align}
\label{eqn:tucker_decomposition_elementwise}
  \widehat{x}_{i_1 i_2 \dots i_N}
  =
  \sum_{r_1=1}^{R_1} \sum_{r_2 = 1}^{R_2} \cdots \sum_{r_N=1}^{R_N}
  g_{r_1 r_2 \dots r_N}
  a_{i_1 r_1}^{(1)} 
  a_{i_2 r_2}^{(2)} 
  \cdots
  a_{i_N r_N}^{(N)}.
\end{align}
Equation~\Cref{eqn:tucker_decomposition_elementwise}
shows that $\widehat{\tensor{X}}$ is the sum of $R_1 R_2 \cdots R_N$ rank-1 tensors.
The tuple $(R_1, R_2, \dots, R_N)$ is the \emph{multilinear rank} of the
decomposition and is chosen to be much smaller than the dimensions of~$\tensor{X}$.
Sometimes columnwise orthogonality constraints are enforced on the factor
matrices, and hence the Tucker decomposition can be thought of as a
higher-order SVD, but such constraints are not required.

\paragraph{Ridge Leverage Scores.}
The \emph{$\lambda$-ridge leverage score} of the $i$-th row of a matrix
$\mat{A} \in \R^{n \times d}$ is
\begin{equation}
\label{eqn:ridge_leverage_score_def}
    \ell_{i}^{\lambda}\parens*{\mat{A}}
    \defeq
    \mat{a}_{i:}\parens*{\mat{A}^\intercal \mat{A} + \lambda\mat{I}}^+ \mat{a}_{i:}^\intercal.
\end{equation}
We also define the related \emph{cross $\lambda$-ridge leverage score} as
$
\ell_{ij}^{\lambda}\parens*{\mat{A}}
\defeq
\mat{a}_{i:}\parens*{\mat{A}^\intercal \mat{A} + \lambda\mat{I}}^+ \mat{a}_{j:}^\intercal$.
The matrix of cross $\lambda$-ridge leverage scores is
$\mat{A}(\mat{A}^\intercal \mat{A} + \lambda\mat{I})^+\mat{A}^\intercal$,
and we denote its diagonal by $\bm{\mat{\ell}}^\lambda(\mat{A})$
since this vector contains the $\lambda$-ridge leverage scores of $\mat{A}$.
Ridge leverage scores generalize the \emph{statistical leverage scores}
of a matrix, in that setting $\lambda = 0$ recovers
the leverage scores of $\mat{A}$, which we denote by
the vector $\bm{\mat{\ell}}(\mat{A})$.
If we let $\mat{A} = \mat{U} \mat{\Sigma}\mat{V}^\intercal$
be the compact SVD of $\mat{A}$, it can be shown that
\begin{align}
\label{eqn:ridge_leverage_score_svd_def}
  \ell_{ij}^\lambda\parens*{\mat{A}}
  =
  \sum_{k=1}^{r}
    \frac{\sigma_{k}^2\parens*{\mat{A}}}{ \sigma_{k}^2\parens*{\mat{A}} + \lambda }
    u_{ik} u_{jk},
\end{align}
where $r = \rank(\mat{A})$.
Therefore, it follows that each $\ell_{i}^\lambda(\mat{A}) \le 1$ since
$\mat{U}$ is an orthogonal matrix.
The \emph{effective dimension} $d_\eff$ of the ridge regression problem
is the sum of $\lambda$-ridge leverage scores:
\begin{align}
\label{eqn:d_eff_def}
  d_\eff
  =
  \sum_{i=1}^n \ell_{i}^\lambda\parens*{\mat{A}}
  =
  \sum_{i=1}^n \sum_{k=1}^r \frac{\sigma_{k}^2\parens*{\mat{A}}}{ \sigma_{k}^2\parens*{\mat{A}} + \lambda} u_{ik}^2
  =
  \sum_{k=1}^r \frac{\sigma_{k}^2\parens*{\mat{A}}}{ \sigma_{k}^2\parens*{\mat{A}} + \lambda}
  \le r.
\end{align}
The regularization parameter $\lambda$ shrinks the dimensionality of the
problem, so $d_\eff = r$ iff $\lambda = 0$.
The $\lambda$-ridge leverage score of a row measures its importance
when constructing the row space of $\mat{A}$ in the context of ridge regression.
See~\cite{ahmed2015fast,cohen2015uniform} for further details and intuition
about $\lambda$-ridge leverage scores.

\section{Approximate Ridge Regression by Leverage Score Sampling}
\label{sec:approximate_ridge_regression}

We start by introducing a row sampling-based approach for approximately solving
any ridge regression problem via $\lambda$-ridge leverage scores.
Computing ridge leverage scores is often as expensive as solving the
ridge regression problem itself, and thus is not immediately useful for
constructing a smaller, sketched instance to use as a proxy.
However, in this section, we show how to use $\lambda$-ridge
leverage scores overestimates to construct a feasible sampling distribution over the
rows of an augmented design matrix,
and then we use this new distribution to efficiently
sketch an ordinary least squares problem whose solution vector is a 
$(1+\varepsilon)$-approximation to the input ridge regression problem.
Furthermore, we show that classical leverage scores are a sufficient
overestimate of $\lambda$-ridge leverage scores,
for any regularization strength $\lambda \ge 0$,
and always result in efficient sketching subroutines.

It will be useful to first give some context into the derivation of
$\lambda$-ridge leverage scores introduced by~\citet{ahmed2015fast}.
Let us rewrite the ridge regression objective as an ordinary least
squares problem in terms of an augmented design matrix and response vector:
\begin{equation}
\label{eqn:problem_def}
  \mat{x}_\opt
  =
  \argmin_{\mat{x} \in \R^{d}}
  \norm*{\mat{A}\mat{x} - \mat{b}}_{2}^2 + \lambda \norm*{\mat{x}}_{2}^2
  =
  \argmin_{\mat{x} \in \R^{d}}
  \norm*{
  \begin{bmatrix}
    \mat{A} \\
    \sqrt{\lambda}\mat{I}_{d}
  \end{bmatrix}
  \mat{x} -
  \begin{bmatrix}
    \mat{b} \\
    \mat{0}
  \end{bmatrix}
  }_{2}^2
  =
  \argmin_{\mat{x} \in \R^{d}}
  \norm*{\overline{\mat{A}}\mat{x} - \overline{\mat{b}}}_{2}^2.
\end{equation}

\begin{restatable}{lemma}{AugmentedDesignMatrixLemma}
\label{lem:augmented_design_matrix}
Let $\mat{A} \in \R^{n \times d}$ be any matrix
and let $\overline{\mat{A}} \in \R^{(n+d)\times d}$
be defined as in \Cref{eqn:problem_def}.
For each row index $i \in [n]$,
the $\lambda$-ridge leverage scores of $\mat{A}$
are equal to the corresponding leverage scores of $\overline{\mat{A}}$.
Concretely, we have $\ell^\lambda_i(\mat{A}) = \ell_i(\overline{\mat{A}})$.
\end{restatable}

\begin{proof}
The lemma follows from 
$\overline{\mat{A}}^\intercal \overline{\mat{A}} = \mat{A}^\intercal \mat{A} + \lambda \mat{I}$
and the definition of $\lambda$-ridge leverage scores.
\end{proof}

\Cref{lem:augmented_design_matrix} gives an alternate explanation
for why the effective dimension $d_\eff$ shrinks as $\lambda$ increases,
since the gap $d - d_\eff$ is the sum of leverage
scores for the $d$ augmented rows corresponding to the regularization terms.

Next we quantify approximate ridge leverage scores from a
probability distribution point of view.
\begin{definition}
The vector $\hat{\bm{\ell}}^\lambda(\mat{A}) \in \R^n$
is a \emph{$\beta$-overestimate} for the $\lambda$-ridge leverage score
distribution of $\mat{A} \in \R^{n \times d}$ if, for all $i \in [n]$,
it satisfies
\[
  \frac{\hat{\ell}_{i}^\lambda \parens*{\mat{A}}}{
    \norm{\hat{\bm{\ell}}^\lambda \parens*{\mat{A}}}_{1} }
  \ge
  \beta 
  \frac{\ell_{i}^\lambda \parens*{\mat{A}}}{
    \norm{\bm{\ell}^\lambda \parens*{\mat{A}}}_{1} }
  =
  \beta
  \frac{\ell_{i}^\lambda \parens*{\mat{A}}}{d_\eff}.
\]
\end{definition}

Whenever $\beta$-overestimates are used in a leverage score-based sketching method,
the sample complexity of the sketching algorithm increases by a factor
of $O(1/\beta)$. Therefore, we want to construct approximate distributions that
minimize the maximum relative decrease in probability.
For ridge leverage score sampling
(in contrast with classical leverage scores),
it is possible to have $\beta > 1$ due to a decrease in effective dimension
$d_\eff$, which is beneficial since it means fewer samples are needed.

Now we define an augmented probability distribution, which will allow us to
seamlessly apply sketching tools for ordinary least squares to ridge
regression.
The augmented distribution is constructed from a ridge leverage score
$\beta$-overestimate, and one of its key properties is that we can 
sample from it in the same amount of time that we can sample
from the $\beta$-overestimate distribution.

\begin{definition}
Let $\mathcal{D}(\hat{\bm{\ell}}^\lambda(\mat{A}), d_\eff')$
denote the \emph{augmented distribution}
of the $\beta$-overestimate $\hat{\bm{\ell}}^\lambda(\mat{A})$,
where $d_\eff' \ge 0$ is a lower bound for the effective dimension of $\mat{A}$.
The sample space of this distribution is the index set $[n+d]$,
and its probability mass function is defined as
  \begin{align*}
    \Pr(X = i) \propto
    \begin{cases}
      \hat{\ell}^\lambda_{i}(\mat{A}) & \text{if $i \in [n]$}, \\
      \min\{1, d - d_\eff'\} & \text{if $i \in [n+d] \setminus [n]$}.
    \end{cases}
  \end{align*}
\end{definition}

Note that sampling from this augmented distribution does not require any
information about the leverage scores for the $d$ additional rows in the
augmented design matrix in Equation~\Cref{eqn:problem_def}.
To efficiently generate a sample from this distribution, we first
flip a coin to branch on the two subsets of indices
and then we sample from each conditional distribution accordingly.

Next we claim the augmented distribution of a
$\beta$-overestimate for a ridge
leverage score distribution is a $\beta'$-overestimate for the
leverage score distribution of the augmented design matrix
in Equation~\Cref{eqn:problem_def}.
The proofs for all of the remaining results in this section
are deferred to \Cref{app:approximate_ridge_regression}.

\begin{restatable}{lemma}{BetaOverestimateTheorem}
\label{lemma:beta_overestimate}
If $\hat{\bm{\ell}}^\lambda(\mat{A})$ is a $\beta$-overestimate for the
$\lambda$-ridge leverage score distribution of $\mat{A} \in \R^{n \times d}$,
then the probability vector for the distribution
$\mathcal{D}(\hat{\bm{\ell}}^\lambda(\mat{A}), d_\eff')$
is a $\beta'$-overestimate for the
leverage score distribution of $\overline{\mat{A}} \in \R^{(n+d) \times d}$,
where
\begin{equation}
\label{eqn:beta_prime}
  \beta' =
  \min\set*{
    \parens*{1 + \frac{d \min\set{1, d-d_\eff'}}{\norm{\hat{\bm{\ell}}^\lambda(\mat{A})}_{1}} }^{-1}
    \frac{\beta d}{d_\eff}
    ,
  \parens*{\frac{\norm{\hat{\bm{\ell}}^\lambda(\mat{A})}_{1}}{d}
  + \min\set{1, d-d_\eff'}}^{-1}
  }.
\end{equation}
\end{restatable}
This lemma is a simple consequence of the definitions of a $\beta$-overestimate
and the effective dimension.
Although Equation~\Cref{eqn:beta_prime} 
may initially seem unwieldy since $d_\eff$ could be difficult to compute, 
we can cleanly lower bound $\beta'$ using 
the fact that $d_\eff \le d$.

The next result follows immediately from the proof of \Cref{lemma:beta_overestimate}.
Observe that there is no explicit dependency on the effective dimension
in this statement, as it gets cancelled out when using classical leverage
scores as the overestimates.

\begin{restatable}{corollary}{BetaOverestimateCorollary}
\label{cor:leverage_score_overestimate}
The leverage scores of $\mat{A} \in \R^{n \times d}$
are a $(d_\eff / \rank(\mat{A}))$-overestimate for the $\lambda$-ridge leverage
scores of $\mat{A}$.
Therefore, the distribution
$\mathcal{D}(\bm{\ell}(\mat{A}), d_\eff')$
is a $(\rank(\mat{A})/d + \min\{1,d - d_\eff'\})^{-1}$-overestimate
for the leverage scores of $\overline{\mat{A}} \in \R^{(n+d) \times d}$.
\end{restatable}
We can always guarantee $\beta' \ge 1/2$ by using the statistical leverage
scores as the $\beta$-overestimate for $\lambda$-ridge leverage scores
since $\rank(\mat{A}) \le d$ and by letting $d_\eff' = 0$.
We choose to write the statements above, however, in terms of $d_\eff'$ since
they perfectly recover leverage score overestimates for ordinary least
squares when $\lambda = 0$ and the effective dimension $d_\eff = d$.

We are now ready to
present our approximate ridge regression algorithm and its guarantees.
This algorithm constructs an augmented distribution from the provided 
$\beta$-overestimate
and uses \Cref{lemma:beta_overestimate}
together with several well-known sketching ideas from randomized
numerical linear algebra~\cite{drineas2006fast,drineas2011faster,woodruff2014sketching}
to bound the number of row samples
needed from the augmented design matrix in Equation~\Cref{eqn:problem_def}.
We explain how all of the leverage score sampling building blocks
interact in~\Cref{app:approximate_ridge_regression}.
In particular, we show that the number of samples $s$ defined on
Line~3 of \Cref{alg:approximate_ridge_regression}
is sufficient for the
sketch matrix $\mat{S}$ to be a subspace embedding of the augmented least
squares problem.

\begin{algorithm}[H]
\caption{Approximate ridge regression using
  a $\beta$-overestimate for $\lambda$-ridge leverage scores.}
\label{alg:approximate_ridge_regression}
\begin{algorithmic}[1]
\Function{\ApproxRidgeRegression}{$\mat{A} \in \R^{n \times d}$, $\mat{b} \in \R^{n}$,
    $\beta$-overestimate 
    $\hat{\bm{\ell}}^\lambda\parens{\mat{A}}$
    for the $\lambda$-ridge
    leverage scores of $\mat{A}$,
    lower bound $d_\eff'$ for the effective dimension of $\mat{A}$,
    $\varepsilon$,
    $\delta$}
  \State Normalize $\hat{\bm{\ell}}^\lambda\parens{\mat{A}}$ so that
    $\norm{\hat{\bm{\ell}}^\lambda\parens{\mat{A}}}_{1} = d$
    and set $\beta' \gets \min\{\beta, 1\} / (1 + \min\{1,d-d_\eff'\})$
  \State Set number of samples $s \gets \ceil{ 4 d / \beta'
      \max\{420 \ln(4d / \delta), 1/(\delta \varepsilon) \}}$ 
  \State Set $\overline{\mat{A}} \gets
  \begin{bmatrix}
    \mat{A}~;~\sqrt{\lambda} \mat{I}_{d}
  \end{bmatrix}$
  and $\overline{\mat{b}} \gets
  \begin{bmatrix}
    \mat{b}~;~\mat{0}
  \end{bmatrix}$
  as in Equation~\Cref{eqn:problem_def}
  \State Initialize sketch matrix $\mat{S} \gets \mat{0}_{s \times (n+d)}$
  \For{$i=1$ to $s$}
    \State Sample $j \sim \mathcal{D}(\hat{\bm{\ell}}^\lambda\parens{\mat{A}},d_\eff')$
    from the augmented distribution
    and set $s_{ij} \gets 1 / \sqrt{\Pr(j) s}$
  \EndFor
  \State \textbf{return} $\mat{\tilde x}_{\opt}
    \gets \argmin_{\mat{x} \in \R^{d}}\norm{\mat{S}\overline{\mat{A}}\mat{x} - \mat{S}\overline{\mat{b}}}_{2}^2$ 
\EndFunction
\end{algorithmic}
\end{algorithm}

\begin{restatable}{theorem}{ApproximateRidgeRegression}
\label{thm:approximate_ridge_regression}
\Cref{alg:approximate_ridge_regression}~samples
$s = O(d\max\{1/\beta,1\} \max\{\ln(d/\delta), 1/(\delta\varepsilon)\})$ rows
and returns a vector $\mat{\tilde x}_{\opt} \in \R^{d}$ such that,
with probability at least $1-\delta$, we have
\[
    \norm*{\mat{A}\mat{\tilde x}_{\opt} - \mat{b}}_{2}^2
    + \lambda \norm*{\mat{\tilde x}_{\opt}}_{2}^2
    \le
    \parens*{1 + \varepsilon}
    \parens*{\norm*{\mat{A}\mat{x}_{\opt} - \mat{b}}_{2}^2
    + \lambda \norm*{\mat{x}_{\opt}}_{2}^2}.
\]
Let $t$ denote the time complexity of sampling from
$\mathcal{D}(\hat{\bm{\ell}}^\lambda(\mat{A}),d_\eff')$
and let $T(n',d')$ be the time needed to solve a least squares
problem of size $n' \times d'$.
The running time of 
\Cref{alg:approximate_ridge_regression} is $O(st + T(s,d))$.
\end{restatable}

One of the main takeaways from this result is that if we use leverage scores
as the input $\beta$-overestimate, \Cref{cor:leverage_score_overestimate}
ensures that only $O(d \max\{\log(d/\delta), 1/(\delta\varepsilon)\})$ row samples
are needed since $\beta' \ge 1/2$.
This holds for any regularization strength $\lambda \ge 0$ since 
$\lambda = 0$ corresponds to full effective dimension and is in some sense
the hardest type of ridge regression problem.
Finally, we note that the running times in \Cref{thm:approximate_ridge_regression}
are templated so that the result can easily be combined with fast sampling
routines and modern least squares algorithms.

\section{Fast Low-Rank Tucker Decomposition}
\label{sec:algorithm}

Now we use the \ApproxRidgeRegression~algorithm to accelerate the
core tensor update steps in the
alternating least squares (ALS) algorithm for Tucker decompositions.
To achieve this, we exploit the Kronecker product structure of the design matrix in
the core tensor subproblem and
use the leverage scores of the design matrix
as an overestimate for the true $\lambda$-ridge leverage scores.
The leverage scores of a Kronecker product matrix factor cleanly
into the product of leverage scores of its factor matrices,
which allows us to sample from the augmented leverage score
distribution in time that is sublinear in the number of rows of
this design matrix.
A similar technique was recently used in the context of ALS for tensor CP decompositions,
where the leverage scores of a Kronecker product matrix were used as an
overestimate for the leverage scores of a Khatri--Rao product
design matrix~\cite{cheng2016spals,larsen2020practical}.

We start by presenting the ALS algorithm for Tucker decompositions
and briefly analyze the amount of work in each step.
To update factor matrix $\mat{A}^{(n)}$, it solves $I_n$
ridge regression problems, all of which share the same design matrix
of size $(I_1 \cdots I_{n-1} I_{n+1} \cdots I_{N}) \times R_{n}$.
For each core tensor update, we solve a ridge regression problem
whose design matrix has dimensions
$(I_1 I_2 \cdots I_{N}) \times (R_1 R_2 \cdots R_N)$.
Updating the core tensor is by far the most expensive step in an
iteration of ALS, hence our motivation for making it faster
by an approximate ridge regression subroutine.

\begin{algorithm}[H]
\caption{Alternating least squares (ALS) algorithm for regularized Tucker decomposition.}
\label{alg:alternating_least_squares}
\begin{algorithmic}[1]
\Function{ALS}{tensor $\tensor{X} \in \R^{I_1\times I_2\times \cdots \times I_N}$,
                    multilinear rank $(R_1,R_2,\dots,R_N)$,
                    regularization $\lambda$}
\State Initialize random core tensor $\tensor{G} \in \R^{R_1 \times R_2 \times \dots \times R_n}$
\State Initialize random factor matrix
$\mat{A}^{(n)} \in \R^{I_n \times R_n}$ for $n=1$ to $N$
    \Repeat
        \For{$n=1$ to $N$}
            \State Set $\mat{K} \gets \mat{G}_{(n)} \parens{\mat{A}^{(1)} \otimes \dots \otimes \mat{A}^{(n-1)} \otimes \mat{A}^{(n+1)} \otimes \dots \otimes \mat{A}^{(N)}}^\intercal$
            and $\mat{B} \gets \mat{X}_{(n)}$
            \For{$i=1$ to $I_{n}$}
              \State Update factor matrix row $\mat{a}^{(n)}_{i:} \gets
                \argmin_{\mat{y} \in \R^{1 \times R_n}}
                    \norm{\mat{y} \mat{K} - \mat{b}_{i:}}_{2}^2
                + \lambda \norm{\mat{y}}_{2}^2$
            \EndFor
        \EndFor
        \State Update $\tensor{G} \gets \argmin_{\tensor{G}'}
             \norm{\parens{\mat{A}^{(1)} \ktimes \mat{A}^{(2)} \ktimes \cdots \ktimes \mat{A}^{(N)}} \vectorize(\tensor{G}') - \vectorize(\tensor{X})}_{2}^2
             + \lambda \norm{\vectorize{\parens{\tensor{G}'}}}_{2}^2$

    \Until{convergence}
    \State \textbf{return}
      $\tensor{G}, \mat{A}^{(1)}, \mat{A}^{(2)}, \dots, \mat{A}^{(N)}$ 
\EndFunction
\end{algorithmic}
\end{algorithm}

\subsection{Approximate Core Tensor Update}

Next we explore the structure of $\lambda$-ridge leverage scores for Kronecker
product matrices and describe how to efficiently sample from the augmented
leverage score distribution.
The following result shows how $\lambda$-ridge leverage scores of a Kronecker
product matrix decompose according to the SVDs of its factor matrices.
In the special case of leverage scores (i.e., $\lambda = 0$),
the expression completely factors,
which we can exploit since it induces a product distribution.
The proof of this result repeatedly uses the mixed-product property
for Kronecker products with the pseudoinverse-based definition
of $\lambda$-ridge leverage scores in Equation~\Cref{eqn:ridge_leverage_score_def}.
We defer the proofs of all results in this section to \Cref{app:algorithm}.

\begin{restatable}{lemma}{KroneckerCrossLeverageScores}
\label{lemma:kronecker_cross_leverage_scores}
Suppose 
$\mat{K} = \mat{A}^{(1)} \otimes \mat{A}^{(2)} \otimes \dots \otimes \mat{A}^{(N)}$,
where each factor matrix $\mat{A}^{(n)} \in \R^{I_n \times R_n}$,
and let $(i_1,i_2,\dots,i_N)$ denote the canonical row indexing of $\mat{K}$
according to its factors.
If the SVD of
  $\mat{A}^{(n)}
  = \mat{U}^{(n)} \mat{\Sigma}^{(n)} {\mat{V}^{(n)}}^\intercal $,
then for $\lambda > 0$, the cross $\lambda$-ridge leverage scores of $\mat{K}$ are
\begin{align*}
    \ell_{(i_1,\dots,i_N),(j_1,\dots,j_N)}^\lambda \parens*{\mat{K}}
    =
    \sum_{\mat{t} \in T}
    \frac{\prod_{n=1}^N \sigma_{t_n}^2(\mat{A}^{(n)})  }{\prod_{n=1}^N \sigma_{t_n}^2(\mat{A}^{(n)}) + \lambda}
    \parens*{
      \prod_{n=1}^N
      u_{i_n t_n}^{(n)}
    }\parens*{
      \prod_{n=1}^N
      u_{j_n t_n}^{(n)}
    },
\end{align*}
where the sum is over the row index set $T = [I_1] \times [I_2] \times \dots \times [I_N]$.
Therefore, given the SVDs of the factor matrices, we can compute each
cross $\lambda$-ridge leverage score of $\mat{K}$ in
$O(R_1 R_2 \cdots R_N \cdot N)$ time.
Furthermore, for (statistical) cross leverage scores, we have
\begin{align*}
  \ell_{(i_1,\dots,i_N),(j_1,\dots,j_N)}\parens*{\mat{K}}
  =
  \prod_{n=1}^N
  \ell_{i_n j_n} \parens*{\mat{A}^{(n)}}.
\end{align*}
\end{restatable}

Now we present our fast sampling-based core tensor update.
The algorithm first computes the leverage scores for each factor matrix, and
then it initializes a data structure to query entries in the Kronecker product
design matrix $\mat{K}$ without explicitly constructing it
(to avoid creating a memory bottleneck).
Similarly, the algorithm then initializes a data structure to sample from the
leverage score distribution of $\mat{K}$ by
independently sampling the index for each dimension
from the corresponding factor matrix leverage score distribution.
This ensures that the time complexity for sampling all row indices of $\mat{K}$
is sublinear in the number of its rows.
Finally, the algorithm calls the approximate ridge regression subroutine with
these data structures as implicit inputs and updates the core tensor
accordingly.

\begin{algorithm}[H]
\caption{Fast core tensor update using approximate ridge regression.}
\label{alg:fast_core_tensor_update}
\begin{algorithmic}[1]
  \Function{FastCoreTensorUpdate}{$\tensor{X} \in \R^{I_1\times I_2\times \cdots \times I_N}$,
    factors $\mat{A}^{(n)} \in \R^{I_n \times R_n}$, $\lambda$,
    $\varepsilon$, $\delta$}
  \State Compute factor matrix leverage scores $\bm{\ell}(\mat{A}^{(n)})$
    by Equation~\Cref{eqn:ridge_leverage_score_def}
    for $n = 1$ to $N$
  \State Initialize data structure to query entries of
    design matrix
    $\mat{K} \gets \mat{A}^{(1)} \otimes \mat{A}^{(2)} \otimes \dots \otimes \mat{A}^{(N)}$ 
  \State Initialize data structure to sample from
    $\bm{\ell}(\mat{K})$ using factored leverage scores
  \State Set core tensor
      $\vectorize(\tensor{G})\hspace{-0.02cm} \gets \hspace{-0.05cm}\ApproxRidgeRegression(\mat{K}, \vectorize\parens{\tensor{X}},\hspace{-0.05cm} \bm{\ell}(\mat{K}), \lambda, 0, \varepsilon, \delta)$
\EndFunction
\end{algorithmic}
\end{algorithm}

\begin{restatable}{theorem}{FastCoreTensorUpdateRunningTime}
\label{thm:fast_core_tensor_update}
Let $R = R_1 R_2 \cdots R_N$.
\Cref{alg:fast_core_tensor_update} gives a $(1+\varepsilon)$-approximation
to the optimal core tensor weights with probability at least $1 - \delta$
in time $O(R^\omega \max\{\log(R/\delta),1/(\delta\varepsilon)\} + \sum_{n=1}^N I_n R_n^2)$,
where $\omega < 2.373$ is the matrix multiplication exponent,
and uses $O(R^2 + \sum_{n=1}^N I_n R_n)$ space.
\end{restatable}

The approximation and running time guarantees in \Cref{thm:fast_core_tensor_update}
follow primarily from \Cref{cor:leverage_score_overestimate}
and \Cref{thm:approximate_ridge_regression}, which show that leverage scores
are a sample-efficient overestimate for any ridge regression problem.
One appealing consequence of an approximate core tensor update whose running time
is predominantly a function of the multilinear rank and sketching parameters
(as opposed to the size of the input tensor $I_1 I_2 \cdots I_N$)
is that the max block improvement (MBI) algorithm~\cite{chen2012maximum}
becomes feasible for a much larger set of problems, 
since ALS is a special case of block coordinate descent.

\subsection{Missing Data and Ridge Leverage Score Upper Bounds}
\label{sec:tensor_completion}

Now we take a step towards the more general tensor
completion problem, where the goal is to learn a tensor decomposition that fits
the observed entries well and also generalizes to unseen data.
For Tucker decompositions and the ALS algorithm, this corresponds to removing
rows of the Kronecker product design matrix $\mat{K}$
(i.e., observations in the input tensor)
before updating the core.
To apply sampling-based sketching techniques in this setting and accelerate
the core tensor update, we need to understand how the $\lambda$-ridge leverage scores
of a matrix change as its rows are removed.
In particular, we need to give accurate upper bounds for the ridge
leverage scores of the matrix with removed rows.
Our first result shows how ridge leverage scores increase as the
rows of the matrix are removed.

\begin{restatable}{theorem}{RidgeScoreUpperBoundAfterRemovingRows}
\label{thm:ridge_score_upper_bound_after_removing_rows}
Let $\mat{L} = \mat{A}\parens{\mat{A}^\intercal\mat{A} + \lambda\mat{I}}^+\mat{A}^\intercal$
be the cross $\lambda$-ridge leverage score matrix of $\mat{A} \in \R^{n \times d}$.
For any $S \subseteq [n]$,
let $\mat{\tilde A}$ denote the matrix containing only the rows of~$\mat{A}$
indexed by $S$.
Let $\overline{S} = [n] \setminus S$ denote the set of missing row indices,
and let $\mat{L}_{\overline{S},\overline{S}}$ be the principal submatrix
of $\mat{L}$ containing only the entries indexed by $\overline{S}$.
Then, for any $\lambda > 0$ and $i \in S$, we have
\begin{align*}
  \ell_{i}^\lambda \parens*{\mat{\tilde A}}
  \le
  \ell_{i}^\lambda \parens*{\mat{A}}
  +
  \frac{1}{1-\lambda_{\max}\parens*{\mat{L}_{\overline{S},\overline{S}}}}
  \sum_{j \in \overline{S}} 
  \ell_{ij}^\lambda (\mat{A})^2,
\end{align*}
where $\lambda_{\max}(\mat{M})$ denotes the maximum eigenvalue of matrix $\mat{M}$.
\end{restatable}

The proof relies on the Woodbury matrix identity and Cauchy interlacing
theorem and provides useful insight into the role
that cross ridge leverage score $\ell^\lambda_{ij}(\mat{A})$
plays when row $j$ is removed but $i$ remains.
We also generalize a well-known property of cross leverage
scores, which allows to sample from a product distribution in the tensor
completion setting when combined with \Cref{thm:ridge_score_upper_bound_after_removing_rows}.
\begin{restatable}{lemma}{SumOfSquaredCrossScores}
For any $\mat{A} \in \R^{n \times d}$ and $i \in [n]$, the
sum of the squared cross $\lambda$-ridge leverage scores is at most
the ridge $\lambda$-leverage score itself, with equality iff $\lambda = 0$.
Concretely, 
$
  \sum_{j=1}^n \ell_{ij}^\lambda \parens*{\mat{A}}^2
  \le
  \ell_{i}^\lambda \parens*{\mat{A}}.
$
\end{restatable}
These results imply that the ridge leverage
scores of $\mat{A}$ are a $\beta$-overestimate for
ridge leverage scores of $\mat{\tilde A}$,
where $1/\beta = 1 + 1/(1 - \lambda_{\max}(\mat{L}_{\overline{S},\overline{S}}))$.
We discuss how this affects the sample complexity for sketching
in~\Cref{app:algorithm}
(e.g.,
if the factor matrices for $\mat{K}$ are semi-orthogonal,
then $1/\beta \le 2 + 1/\lambda$).

\section{Experiments}
\label{sec:experiments}

We compare the performance of ALS with and without
the \textsc{FastCoreUpdate} algorithm
by computing low-rank Tucker decompositions
of
large, dense synthetic tensors
and real-world movie data from the \textsc{TensorSketch} experiments of
\citet{malik2018low}.
Our experiments build on Tensorly~\cite{kossaifi2019tensorly}
and were run on an
Intel Xeon W-2135 processor (8.25 MB cache and 3.70 GHz)
using 128GB of RAM.
We refer to the ALS algorithm that uses ridge leverage score sampling
for the core as ALS-RS.

\paragraph{Synthetic Benchmarks.}
First we stress test ALS and ALS-RS on large, dense tensors.
We consider various low-rank Tucker decompositions
and show that the running time to update the core in ALS-RS remains
fixed as a function of the core size, without any loss in solution accuracy.
These profiles further illustrate that the core tensor update is
a severe bottleneck in an unmodified ALS algorithm.

Specifically, for each input shape we generate a random Tucker
decomposition with an (8, 8, 8) core tensor.
All of the entries in the factor matrices and core tensor are i.i.d.\
uniform random variables from $[0,1]$.
We add Gaussian noise from $\mathcal{N}(0,1)$ to one percent of
the entries in this Tucker tensor,
and call the resulting tensor $\tensor{Y}$.
The noise gives a lower bound of 0.10 for the RMSE when learning~$\tensor{Y}$.
Next we initialize a random Tucker decomposition with a smaller
core size and fit it to $\tensor{Y}$ using ALS and ALS-RS,
both starting from the same initial state.
For each input shape and multilinear rank of the learned Tucker decomposition,
we run ALS and ALS-RS for 10 iterations and report the mean
running time of each step type in \Cref{table:synthetic-benchmarks}.
In all the experiments we set $\lambda = 0.001$, and for
ALS-RS we set $\varepsilon = 0.1$ and $\delta = 0.1$.
The columns labeled F1, F2, and F3 correspond to factor matrix updates.

{
\setlength{\tabcolsep}{5.25pt}
\begin{table}
  \caption{Mean running times per step type for ALS and ALS-RS on dense synthetic tensors.}
  \label{table:synthetic-benchmarks}
  \centering
  \begin{tabular}{llrrrrrrrrrrrrrrrrrrrrrrrr}
    \toprule
    & & & & & \multicolumn{2}{c}{ALS} & \multicolumn{2}{c}{ALS-RS} \\
    \cmidrule(r){6-7} 
    \cmidrule(r){8-9} 
    Input Shape & Rank & F1 (s) & F2 (s) & F3 (s) & Core (s) & RMSE & Core (s) & RMSE \\
    \midrule
    (512, 512, 512) & (2, 2, 2) & 2.0	& 2.0	& 0.3	& 7.5 & 0.364 & 3.9 & 0.364 \\
                & (4, 2, 2) & 2.1	& 2.0	& 0.3 & 13.3 & 0.363 & 9.4 & 0.363 \\
                & (4, 4, 2) & 2.1	& 2.1	& 0.3	& 24.8 & 0.328 & 23.1 & 0.328 \\
                & (4, 4, 4) & 2.1	& 2.1	& 0.4	& 48.1 & 0.284 & 54.7 & 0.284 \\
    \midrule
    (1024, 512, 512) & (2, 2, 2) & 4.5 & 4.4 & 0.6 & 15.4 & 0.392 & 3.9 & 0.392 \\
                & (4, 2, 2) & 4.6 & 4.4 & 0.7 & 26.9 & 0.387 & 9.3 & 0.387 \\
                & (4, 4, 2) & 4.7 & 4.6 & 0.7 & 50.3 & 0.342 & 22.7 & 0.342 \\
                & (4, 4, 4) & 4.7	& 4.6	& 0.9	& 96.8 & 0.300 & 53.9 & 0.300 \\
    \midrule
    (1024, 1024, 512) & (2, 2, 2) & 13.3 & 13.2 & 1.3 & 35.0 & 0.425 & 3.9 & 0.425 \\
                & (4, 2, 2) & 13.4 & 13.0 & 1.3 & 58.0 & 0.413 & 9.3 & 0.413 \\
                & (4, 4, 2) & 13.7 & 13.7	& 1.4	& 104.0 & 0.382 & 22.7 & 0.382 \\
                & (4, 4, 4) & 13.7 & 13.7	& 1.9	& 196.9 & 0.321 & 54.0 & 0.321 \\
    \midrule
    (1024, 1024, 1024) & (2, 2, 2) & 27.7 & 27.6 & 2.7 & 67.9 & 0.406 & 3.9 & 0.406 \\
                   & (4, 2, 2) & 28.8	& 27.6 & 2.7 & 110.9 & 0.403 & 9.3 & 0.403 \\
                   & (4, 4, 2) & 28.9	& 28.8 & 2.8 & 196.3 & 0.355 & 22.8 & 0.356 \\
                   & (4, 4, 4) & 28.5	& 28.3 & 3.8 & 367.1 & 0.294 & 54.0 & 0.294 \\
    \bottomrule
  \end{tabular}
\end{table}
}

We draw several conclusions from~\Cref{table:synthetic-benchmarks}.
First, by using ridge leverage score sampling, ALS-RS
guarantees fast, high-quality core tensor updates while solving a substantially
smaller ridge regression problem in each iteration.
Comparing the core columns of ALS and ALS-RS at fixed ranks as the input
shape increases clearly demonstrates this.
Second, ALS and ALS-RS converge to the same local optimum for this data,
which is evident from the RMSE columns since both start from the same seed.

\paragraph{Movie Experiments.} 
Now we use ALS and ALS-RS to compute low-rank Tucker decompositions of a video
of a nature scene.
\citet{malik2018low} used their 
\textsc{TensorSketch} algorithm for Tucker decompositions
on this data since most frames are very similar,
except for a couple disturbances when a person walks by.
Thus, the video is expected to be compressible.
It consists of 2493 frames, each an image of size (1080, 1920, 3)
corresponding to height, width, and the RGB channel.
Similar to the experiments in \cite{malik2018low}, we
construct tensors of shape (100, 1080, 1920, 3) from 100-frame slices of the
video for our comparison.
We again let $\lambda = 0.001$, and set
$\varepsilon = 0.1$ and $\delta = 0.1$ for ALS-RS.

We use a Tucker decomposition with a core of shape (2, 4, 4, 2),
and we present the mean running times for each update step type over 5
iterations, after which both algorithms have converged.
The update steps of ALS took
(8.8, 3.1, 8.8, 163.1, 9935.5)
seconds,
corresponding to F1, F2, F3, F4, and the core update.
In contrast, ALS-RS took
(9.2, 2.7, 9.3, 221.6, 85.8) seconds per step type.
Both algorithms converge to 0.188 RMSE,
and at each step of every iteration,
the RMSEs of the two solutions differed in absolute value by at most 2.98e-5.
Hence, they converged along the same paths.
Using ALS-RS in this experiment sped up the core tensor update by more
than a factor of 100 and shifted the performance bottleneck to the
factor matrix update corresponding to the RGB channel.

\section{Conclusion}
\label{sec:conclusion}

This work accelerates ALS for regularized low-rank Tucker
decompositions by using ridge leverage score sampling in the core
tensor update.
Specificaly, it introduces a new approximate ridge regression algorithm that
augments classical leverage score distributions,
and it begins to explore properties of ridge leverage scores in a dynamic
setting.
These results also immediately extend the leverage score sampling-based CP
decomposition algorithms in~\cite{cheng2016spals,larsen2020practical} to
support L2 regularization.
Lastly, our synthetic and real-world experiments demonstrate that this
approximate core tensor update leads to substantial improvements in the running
time of ALS for low-rank Tucker decompositions while preserving the original
solution quality of ALS.

\section*{Acknowledgements}
MG was 
supported in part by the
Institute for Data Engineering and Science (IDEaS) and
Transdisciplinary Research Institute for Advancing Data Science (TRIAD)
Research Scholarships for Ph.D.\ students and postdocs.
Part of this work was done while he was a summer intern at Google Research.

\bibliographystyle{plainnat}
\bibliography{references}

\appendix

\newpage
\section{Missing Analysis from \Cref{sec:approximate_ridge_regression}}
\label{app:approximate_ridge_regression}

Here we show how to use $\beta$-overestimates for the
ridge leverage scores of a design matrix $\mat{A} \in \R^{n \times d}$
to create a substantially smaller ordinary least squares problem
whose solution vector gives a $(1+\varepsilon)$-approximation to the
original ridge regression problem.
Our proof relies on several sketching and leverage score sampling
results from randomized numerical linear
algebra~\cite{drineas2006fast,drineas2011faster,woodruff2014sketching}.
The leverage score sampling part of our argument was recently
organized all in one place in \cite[Appendix B]{larsen2020practical},
as these prerequisite results are well-understood but
scattered across many references.

\subsection{Fast Approximate Least Squares}
\label{subsec:fast_least_squares}
We start by considering the overdetermined least squares problem
defined by a matrix $\mat{A} \in \R^{n \times d}$ and
response vector $\mat{b} \in \R^{n}$,
where $n \ge d$ and $\rank(\mat{A}) = d$.
Define the optimal sum of squared residuals to be
\begin{equation}
\label{eqn:least_squares_problem}
  \cR^2
  =
  \min_{x \in \R^{d}}\norm*{\mat{A}\mat{x} - \mat{b}}_{2}^2.
\end{equation}
Assume for now that $\mat{A}$ is full rank, and
let the compact SVD of the design matrix be
$\mat{A} = \mat{U}_{\mat{A}} \mat{\Sigma}_{\mat{A}} \mat{V}_{\mat{A}}^\intercal$.
By definition, $\mat{U}_{\mat{A}} \in \R^{n \times d}$ is an orthonormal
basis for the column space of $\mat{A}$.
Let $\mat{U}_{\mat{A}}^\perp \in \R^{n \times (n-d)}$ be an orthonormal
basis for the $(n-d)$-dimensional subspace that is
orthogonal to the column space of~$\mat{A}$.
For notational simplicity,
let $\mat{b}^\perp = \mat{U}_{\mat{A}}^\perp {\mat{U}_{\mat{A}}^\perp}^\intercal \mat{b}$
denote the projection of $\mat{b}$ onto the orthogonal subspace $\mat{U}_{\mat{A}}^\perp$.
The vector $\mat{b}^\intercal$
is important because its norm is equal to the norm of the residual vector.
To see this, observe that $\mat{x}$ can be chosen so that $\mat{A}\mat{x}$
perfectly matches the part of $\mat{b}$ in the column space of $\mat{A}$, but
cannot (by definition) match anything in the range of $\mat{U}_{\mat{A}}^\perp$:
\[
  \cR^2
  =
  \min_{x \in \R^{d}}\norm*{\mat{A}\mat{x} - \mat{b}}_{2}^2
  =
  \norm*{\mat{U}_{\mat{A}}^\perp {\mat{U}_{\mat{A}}^\perp}^\intercal \mat{b}}_{2}^{2}
  =
  \norm*{\mat{b}^\perp}_{2}^{2}.
\]
We denote the solution to the least squares problem by $\mat{x}_{\opt}$,
hence we have $\mat{b} = \mat{A}\mat{x}_{\opt} + \mat{b}^\perp$.

Now we build on a structural result of \citet[Lemma 1]{drineas2011faster}
that establishes sufficient conditions on any sketching
matrix $\mat{S} \in \R^{s \times n}$
such that the solution $\mat{\tilde x}_{\opt}$ to the approximate
least squares problem
\begin{align}
\label{eqn:approximate_least_squares_problem}
  \mat{\tilde x}_{\opt}
  =
  \argmin_{x \in \R^d}\norm*{\mat{S} \parens*{\mat{A}\mat{x} - \mat{b}}}_{2}^2
\end{align}
gives a relative-error approximation to the original least squares problem.
The two conditions that we require of matrix $\mat{S}$ are:
\begin{align}
\label{eqn:structural_condition_1}
  &\sigma^2_{\min}\parens*{\mat{S} \mat{U}_{\mat{A}}} \ge 1 / \sqrt{2}, \text{~and~} \\
\label{eqn:structural_condition_2}
  &\norm*{\mat{U}_{\mat{A}}^\intercal \mat{S}^\intercal \mat{S} \mat{b}^\perp }_{2}^2 \le \varepsilon \cR^2 / 2,
\end{align}
for some $\varepsilon \in (0,1)$.
Note that while the \ApproxRidgeRegression~algorithm is randomized,
the following lemma is a deterministic statement.
Failure probabilities will enter our analysis later when we show that
this algorithm satisfies conditions~\Cref{eqn:structural_condition_1}
and \Cref{eqn:structural_condition_2} with sufficiently high probability.

\begin{lemma}[\cite{drineas2011faster}]
\label{lem:approx_least_squares}
Consider the overconstrained least squares approximation
problem in~\Cref{eqn:least_squares_problem}, and
let the matrix $\mat{U}_{\mat{A}} \in \R^{n \times d}$ contain the top $d$
left singular vectors of $\mat{A}$.
Assume the matrix $\mat{S}$ satisfies conditions
\Cref{eqn:structural_condition_1} and
\Cref{eqn:structural_condition_2} above, for some $\varepsilon \in (0, 1)$.
Then, the solution $\mat{\tilde x}_{\opt}$ to the approximate least squares
problem~\Cref{eqn:approximate_least_squares_problem} satisfies:
\begin{align}
  \norm*{\mat{A} \mat{\tilde x}_{\opt} - \mat{b} }_{2}^2
  &\le
  \parens*{1 + \varepsilon} \cR^2, \text{~and~}
  \\
  \norm*{\mat{\tilde x}_{\opt} - \mat{x}_{\opt}}_{2}^2
  &\le
  \frac{1}{\sigma^2_{\min}\parens*{\mat{A}}}
  \varepsilon \cR^2.
\end{align}
\end{lemma}

\begin{proof}
Let us first rewrite the sketched least squares problem induced by $\mat{S}$ as
\begin{align}
    \label{eqn:sketched_ls_proof_1}
    \min_{\mat{x} \in \R^{d}} \norm*{\mat{S} \mat{A}\mat{x} - \mat{S}\mat{b}}_{2}^2
    &=
    \min_{\mat{y} \in \R^{d}} \norm*{\mat{S} \mat{A}\parens*{\mat{x}_{\opt} + \mat{y}} - \mat{S}\parens*{\mat{A}\mat{x}_{\opt} + \mat{b}^{\perp}}}_{2}^2 \\
    &=
    \min_{\mat{y} \in \R^{d}} \norm*{\mat{S} \mat{A}\mat{y} - \mat{S}\mat{b}^{\perp}}_{2}^2 \notag \\
    \label{eqn:sketched_ls_proof_2}
    &=
    \min_{\mat{z} \in \R^{d}} \norm*{\mat{S} \mat{U}_{\mat{A}}\mat{z} - \mat{S}\mat{b}^{\perp}}_{2}^2.
\end{align}
Equation~\Cref{eqn:sketched_ls_proof_1} follows
since $\mat{b} = \mat{A}\mat{x}_{\opt} + \mat{b}^\perp$,
and \Cref{eqn:sketched_ls_proof_2} follows because the columns of $\mat{A}$ span
the same subspace as the columns of $\mat{U}_{\mat{A}}$.
Now, let $\mat{z}_{\opt} \in \R^{d}$ be such that
$\mat{U}_{\mat{A}} \mat{z}_{\opt} = \mat{A}(\mat{\tilde x}_{\opt} - \mat{x}_{\opt})$,
and note that $\mat{z}_{\opt}$ minimizes \Cref{eqn:sketched_ls_proof_2}.
This fact follows from
\begin{align*}
    \norm*{\mat{S}\mat{A}(\mat{\tilde x}_{\opt} - \mat{x}_{\opt}) - \mat{S}\mat{b}^\perp}_{2}^2
    =
    \norm*{
    \mat{S}\mat{A}\mat{\tilde x}_{\opt}
    - \mat{S}\parens*{\mat{b} - \mat{b}^\perp} - \mat{S}\mat{b}^\perp }_{2}^2
    =
    \norm*{\mat{S}\mat{A}\mat{\tilde x}_{\opt} - \mat{S}\mat{b}}_{2}^2.
\end{align*}
Thus, by the normal equations, we have
\[
    \parens*{\mat{S}\mat{U}_{\mat{A}}}^\intercal \mat{S}\mat{U}_{\mat{A}} \mat{z}_{\opt}
    =
    \parens*{\mat{S}\mat{U}_{\mat{A}}}^\intercal \mat{S} \mat{b}^\perp.
\]
Taking the norm of both sides and observing that under condition
\Cref{eqn:structural_condition_1}
we have
$\sigma_{i}((\mat{S}\mat{U}_{\mat{A}})^\intercal \mat{S}\mat{U}_{\mat{A}})
=
\sigma_{i}^2 (\mat{S}\mat{U}_{\mat{A}}) \ge 1/\sqrt{2}$, for all $i$,
it follows that
\begin{equation}
  \label{eqn:sketched_ls_proof_3}
    \norm*{\mat{z}_{\opt}}_{2}^2 / 2
    \le
    \norm*{\parens*{\mat{S}\mat{U}_{\mat{A}}}^\intercal \mat{S}\mat{U}_{\mat{A}} \mat{z}_{\opt}}_{2}^2
    =
    \norm*{\parens*{\mat{S}\mat{U}_{\mat{A}}}^\intercal \mat{S}\mat{b}^\perp }_{2}^2.
\end{equation}
Using condition \Cref{eqn:structural_condition_2}, we observe that
\begin{equation}
  \label{eqn:sketched_ls_proof_4}
    \norm*{\mat{z}_{\opt}}_{2}^2
    \le
    \varepsilon \cR^2.
\end{equation}

To establish the first claim of the lemma, let us rewrite the squared
norm of the residual vector as
\begin{align}
    \norm*{\mat{A} \mat{\tilde x}_{\opt} - \mat{b}}_{2}^2
    &=
    \norm*{\mat{A} \mat{\tilde x}_{\opt} - \mat{A}\mat{x}_{\opt} + \mat{A}\mat{x}_{\opt} - \mat{b}}_{2}^2 \notag \\
    &=
    \label{eqn:sketched_ls_proof_5}
    \norm*{\mat{A} \mat{\tilde x}_{\opt} - \mat{A}\mat{x}_{\opt}}_{2}^2
    + \norm*{\mat{A}\mat{x}_{\opt} - \mat{b}}_{2}^2 \\
    &=
    \label{eqn:sketched_ls_proof_6}
    \norm*{\mat{U}_{\mat{A}} \mat{z}_{\opt}}_{2}^2
    + \cR^2 \\
    &\le
    \label{eqn:sketched_ls_proof_7}
    \parens*{1 + \varepsilon} \cR^2,
\end{align}
where~\Cref{eqn:sketched_ls_proof_5} follows from the Pythagorean theorem
since $\mat{b}-\mat{A}\mat{x}_{\opt} = \mat{b}^\perp$, which is orthogonal
to $\mat{A}$, and consequently $\mat{A}(\mat{x}_{\opt} - \mat{\tilde x}_{\opt})$;
\Cref{eqn:sketched_ls_proof_6} follows from the definition of $\mat{z}_{\opt}$
and $\cR^2$;
and \Cref{eqn:sketched_ls_proof_7} follows from
\Cref{eqn:sketched_ls_proof_4} and the orthogonality of $\mat{U}_{\mat{A}}$.

To establish the second claim of the lemma, recall that
$\mat{A}(\mat{x}_{\opt} - \mat{\tilde x}_{\opt}) = \mat{U}_{\mat{A}} \mat{z}_{\opt}$.
Taking the norm of both sides of this expression, we have that
\begin{align}
    \label{eqn:sketched_ls_proof_8}
    \norm*{\mat{x}_{\opt} - \mat{\tilde x}_{\opt}}_{2}^2
    &\le 
    \frac{\norm*{\mat{U}_{\mat{A}} \mat{z}_{\opt}}_{2}^2}{\sigma_{\min}^2\parens*{\mat{A}}} \\
    \label{eqn:sketched_ls_proof_9}
    &\le
    \frac{\varepsilon \cR^2 }{\sigma_{\min}^2\parens*{\mat{A}}},
\end{align}
where \Cref{eqn:sketched_ls_proof_8} follows since $\sigma_{\min}(\mat{A})$
is the smallest singular
value of $\mat{A}$ and $\rank(\mat{A})=d$;
and \Cref{eqn:sketched_ls_proof_9} follows from
\Cref{eqn:sketched_ls_proof_4} and the orthogonality of $\mat{U}_{\mat{A}}$.
\end{proof}

\subsection{Tools for Satisfying the Structural Conditions}

Next we present two results that will be useful for proving that the
sketching matrix $\mat{S}$ satisfies the structural conditions in
Equations~\Cref{eqn:structural_condition_1} and \Cref{eqn:structural_condition_2}.
The first result is \cite[Theorem 17]{woodruff2014sketching},
which states that $\mat{S}\mat{U}_{\mat{A}}$ is a subspace embedding
for the column space of $\mat{U}_{\mat{A}}$.
This result can be thought of as approximate isometry
and is noticeably stronger than the desired condition
$\sigma_{\min}^2(\mat{S}\mat{U}_{\mat{A}}) \ge 1/\sqrt{2}$.

\begin{theorem}[\cite{woodruff2014sketching}]
\label{thm:structural_tool_1}
Consider $\mat{A} \in \R^{n \times d}$ and its compact SVD 
$\mat{A} = \mat{U}_{\mat{A}} \mat{\Sigma}_{\mat{A}} \mat{V}_{\mat{A}}^\intercal$.
For any $\beta > 0$, let $\hat{\bm{\ell}}(\mat{A})$
be a $\beta$-overestimate for the
leverage score distribution of $\mat{A}$.
Let $s > 144d \ln(2d / \delta) / (\beta \varepsilon^2)$.
Construct a sampling matrix $\mat{\Omega} \in \R^{n \times s}$ 
and rescaling matrix $\mat{D} \in \R^{s \times s}$ as follows.
Initially, let $\mat{\Omega} = \mat{0}^{n \times s}$
and $\mat{D} = \mat{0}^{s \times s}$.
For each column $j$ of $\mat{\Omega}$ and $\mat{D}$, independently,
and with replacement, choose a row index $i \in [n]$ with probability
$q_i = \hat{\ell}_{i}(\mat{A}) / \norm{\hat{\bm{\ell}}(\mat{A})}_{1}$,
and set $\omega_{ij} = 1$ and $d_{jj} = 1 / \sqrt{q_i s}$.
Then with probability at least $1-\delta$, simultaneously for all $i$,
we have
\[
  1 - \varepsilon
  \le
  \sigma_{i}^2 \parens*{\mat{D}^\intercal \mat{\Omega}^\intercal \mat{U}_{\mat{A}}}
  \le
  1 + \varepsilon.
\]
\end{theorem}

To prove the second structural condition, we use the following result about
squared-distance sampling for approximate matrix multiplication
in~\cite[Lemma 8]{drineas2006fast}.
In our analysis for leverage score sampling-based ridge regression,
it is possible (and beneficial) that $\beta > 1$ in the statement below.
Therefore, we modify the original theorem statement and provide a proof to show
that the result is unaffected.

\begin{theorem}[\cite{drineas2006fast}]
\label{thm:structural_tool_2}
Let $\mat{A} \in \R^{n \times m}$, $\mat{B} \in \R^{n \times p}$,
and $s$ denote the number of samples.
Let $\mat{p} \in \R^{n}$ be a probability vector such that,
for all $i \in [n]$, we have
\[
  p_{i} \ge \beta \frac{\norm{\mat{a}_{i:}}_{2}^2}{ \norm*{\mat{A}}_{\frobenius}^2 },
\]
for some constant $\beta > 0$.
Sample $s$ row indices $(\xi^{(1)},\xi^{(2)},\dots,\xi^{(s)})$ from $\mat{p}$,
independently and with replacement, and form the approximate product
\[
  \frac{1}{s} \sum_{t=1}^s \frac{1}{p_{\xi^{(t)}}}
    \mat{a}_{\xi^{(t)}:}^\intercal \mat{b}_{\xi^{(t)}:}
  =
  \parens*{\mat{S}\mat{A}}^\intercal \mat{S}\mat{B},
\]
where $\mat{S} \in \R^{s \times n}$
is the sampling and rescaling matrix whose $t$-th row is defined by
the entries
\[
  s_{tk} =
  \begin{cases}
    \frac{1}{\sqrt{s p_{k}}} & \text{if $k = \xi_{t}$}, \\
    0 & \text{otherwise}.
  \end{cases}
\]
Then, we have
\[
  \E\bracks*{\norm*{\mat{A}^\intercal \mat{B}
  - \parens*{\mat{S}\mat{A}}^\intercal\mat{S}\mat{B}}_{\frobenius}^2}
  \le
  \frac{1}{\beta s}
  \norm*{\mat{A}}_{\frobenius}^2
  \norm*{\mat{B}}_{\frobenius}^2.
\]
\end{theorem}

\begin{proof}
First, we analyze the entry of 
$(\mat{S}\mat{A})^\intercal \mat{S}\mat{B}$ at index $(i,j)$.
By viewing the approximate product as a sum of outer products,
we can write this entry in terms of scalar random variables
$X_t$, for $t \in [s]$, as follows:
\[
  X_t = \frac{a_{\xi^{(t)}i} b_{\xi^{(t)}j}}{s p_{\xi^{(t)}}}
  \implies
  \bracks*{(\mat{S}\mat{A})^\intercal \mat{S}\mat{B}}_{ij}
  = \sum_{t=1}^s X_t.
\]
The expected values of $X_t$ and $X_t^2$ for all values of $t$ are
\begin{align*}
  \E\bracks*{X_t} &= \sum_{k=1}^n p_k \frac{a_{ki} b_{kj}}{s p_k}
                   = \frac{1}{s} \parens*{\mat{A}^\intercal \mat{B}}_{ij}, \text{~and~} \\
  \E\bracks*{X_t^2} &=
        \sum_{k=1}^n p_k \parens*{\frac{a_{ki} b_{kj}}{s p_k}}^2
        =
        \frac{1}{s^2}\sum_{k=1}^n \frac{\parens*{a_{ki} b_{kj}}^2}{p_k}.
\end{align*}
Thus, we have
$\E[\bracks{(\mat{S}\mat{A})^\intercal \mat{S}\mat{B}}_{ij}]
  = \sum_{t=1}^s \E[X_t] = \parens{\mat{A}^\intercal \mat{B}}_{ij}$,
which means the estimator is unbiased.
Furthermore, since the estimated matrix entry is the sum of $s$
i.i.d.\ random variables,
its variance is
\begin{align*}
  \var\parens*{\bracks*{(\mat{S}\mat{A})^\intercal \mat{S}\mat{B}}_{ij}}
  &=
  \sum_{t=1}^s \var\parens*{X_t} \\
  &=
  \sum_{t=1}^s \parens*{\E\bracks*{X_t^2} - \E\bracks*{X_t}^2} \\
  &=
  \sum_{t=1}^s \frac{1}{s^2} \sum_{k=1}^n \parens*{\frac{\parens*{a_{ki} b_{kj}}^2}{p_k} 
        - \parens*{ \mat{A}^\intercal\mat{B} }_{ij}^2 } \\
  &=
  \frac{1}{s} \sum_{k=1}^n \parens*{\frac{\parens*{a_{ki} b_{kj}}^2}{p_k} 
        - \parens*{ \mat{A}^\intercal\mat{B} }_{ij}^2 }.
\end{align*}
Now we apply this result to the expectation we want to bound:
\begin{align*}
  \E\bracks*{\norm*{\mat{A}^\intercal \mat{B}
  - \parens*{\mat{S}\mat{A}}^\intercal\mat{S}\mat{B}}_{\frobenius}^2}
  &=
  \sum_{i=1}^m \sum_{j=1}^p
  \E\bracks*{
    \parens*{ \bracks*{\parens*{\mat{S}\mat{A}}^\intercal\mat{S}\mat{B}}_{ij}
    - \parens*{\mat{A}^\intercal \mat{B}}_{ij}}^2} \\
  &=
  \sum_{i=1}^m \sum_{j=1}^p
  \E\bracks*{
    \parens*{ \bracks*{\parens*{\mat{S}\mat{A}}^\intercal\mat{S}\mat{B}}_{ij}
    - \E\bracks*{ \bracks*{\parens*{\mat{S}\mat{A}}^\intercal\mat{S}\mat{B}}_{ij} }}^2} \\
  &=
  \sum_{i=1}^m \sum_{j=1}^p
  \var\parens*{
    \bracks*{\parens*{\mat{S}\mat{A}}^\intercal\mat{S}\mat{B}}_{ij}} \\
  &=
  \frac{1}{s} \sum_{i=1}^m \sum_{j=1}^p
  \sum_{k=1}^n \parens*{\frac{\parens*{a_{ki} b_{kj}}^2}{p_k} 
        - \parens*{ \mat{A}^\intercal\mat{B} }_{ij}^2 } \\
  &=
  \frac{1}{s} 
  \sum_{k=1}^n \frac{\parens{\sum_{i=1}^m a_{ki}^2 }\parens{\sum_{j=1}^p b_{kj}^2 }}{p_k} 
        - \frac{n}{s}\sum_{i=1}^m\sum_{j=1}^p \parens*{ \mat{A}^\intercal\mat{B} }_{ij}^2 \\
  &=
  \frac{1}{s} 
  \sum_{k=1}^n
  \frac{\norm{\mat{a}_{k:}}_{2}^2 \norm{\mat{b}_{k:}}_{2}^2}{p_k}
  - \frac{n}{s} \norm*{\mat{A}^\intercal \mat{B}}_{\frobenius}^2 \\
  &\le
  \frac{1}{s} 
  \sum_{k=1}^n
  \frac{\norm{\mat{a}_{k:}}_{2}^2 \norm{\mat{b}_{k:}}_{2}^2}{p_k},
\end{align*}
where the last inequality uses the fact that the Frobenius norm of any matrix
is nonnegative.
Finally, by using the $\beta$-overestimate assumption on the sampling
probabilities, we have
\begin{align*}
  \E\bracks*{\norm*{\mat{A}^\intercal \mat{B}
  - \parens*{\mat{S}\mat{A}}^\intercal\mat{S}\mat{B}}_{\frobenius}^2}
  &\le 
  \frac{1}{s} 
  \sum_{k=1}^n
  \frac{\norm{\mat{a}_{k:}}_{2}^2 \norm{\mat{b}_{k:}}_{2}^2}{p_k} \\
  &\le
  \frac{1}{s} 
  \sum_{k=1}^n
  \parens*{ \norm*{\mat{A}}_{\frobenius}^2  \frac{\norm{\mat{a}_{k:}}_{2}^2 \norm{\mat{b}_{k:}}_{2}^2}{\beta \norm{\mat{a}_{k:}}_{2}^2} } \\
  &=
  \frac{1}{s \beta}
  \norm*{\mat{A}}_{\frobenius}^2
  \sum_{k=1}^n \norm*{\mat{b}_{k:}}_{2}^2 \\
  &=
  \frac{1}{s \beta}
  \norm*{\mat{A}}_{\frobenius}^2
  \norm*{\mat{B}}_{\frobenius}^2,
\end{align*}
which is the desired upper bound.
\end{proof}

\subsection{Generalizing to Ridge Regression}

Now that all of our tools are in place, we show how to design an augmented
distribution that allows us to efficiently use $\beta$-overestimates for the
$\lambda$-ridge leverage scores of $\mat{A}$ to construct a sketch of the
ridge regression problem whose solution gives a
$(1+\varepsilon)$-approximation for the original loss function.
Moreover, our analysis carefully tracks how $\beta'$ (the overestimate factor
for the augmented distribution) changes as a function of $\beta$ and the
effective dimension $d_\eff$ of the original problem.

\BetaOverestimateTheorem*

\begin{proof}
We start by analyzing the leverage score distribution for $\overline{\mat{A}}$.
For row indices $i \in [n]$, we have
\begin{align*}
  \Pr(X = i) = \frac{\ell_{i}\parens*{\overline{\mat{A}}}}{d}.
\end{align*}
It follows from \Cref{lem:augmented_design_matrix} and the definition of
effective dimension that
$\sum_{i=1}^n \ell_{i}\parens*{\overline{\mat{A}}} = d_\eff$.
Hence, for any augmented row index $j \in \{n+1,n+2,\dots,n+d\}$, we have
\begin{align}
\label{eqn:beta_bound_1}
  \Pr(X = j) \le \frac{\min\set*{1, d - d_\eff}}{d}.
\end{align}
Next, we analyze the distribution $\cD(\hat{\bm{\ell}}^\lambda(\mat{A}) , d_\eff')$.
Each row index $i \in [n]$ is sampled with probability
\begin{align}
  \label{eqn:beta_bound_2}
  \Pr(Y = i) &= \frac{\hat{\ell}_{i}^\lambda (\mat{A})}{\norm{\hat{\bm{\ell}}^\lambda (\mat{A})}_{1} + d\min\set{1,d - d_\eff'}} \\
  &= \frac{ \norm{\hat{\bm{\ell}}^\lambda (\mat{A})}_{1} }{ \norm{\hat{\bm{\ell}}^\lambda (\mat{A})}_{1} + d\min\set{1,d - d_\eff'}} \cdot \frac{ \hat{\ell}_{i}^\lambda (\mat{A}) }{ \norm{\hat{\bm{\ell}}^\lambda (\mat{A})}_{1} } \notag \\
  \label{eqn:beta_bound_3}
  &\ge \frac{ \norm{\hat{\bm{\ell}}^\lambda (\mat{A})}_{1} }{ \norm{\hat{\bm{\ell}}^\lambda (\mat{A})}_{1} + d\min\set{1,d - d_\eff'} } \cdot \beta \frac{ \ell_{i}^\lambda (\mat{A}) }{ \norm{\bm{\ell}^\lambda (\mat{A})}_{1} } \\
  \label{eqn:beta_bound_4}
  &= \parens*{1 + \frac{d\min\set{1,d - d_\eff'}}{\norm{\hat{\bm{\ell}}^\lambda (\mat{A})}_{1}}}^{-1} \beta \frac{ \ell_{i}\parens*{\overline{\mat{A}}} }{ d_\eff } \\
  &= \parens*{1 + \frac{d\min\set{1,d - d_\eff'}}{\norm{\hat{\bm{\ell}}^\lambda (\mat{A})}_{1}}}^{-1} \frac{\beta d}{d_\eff} \Pr(X = i). \notag
\end{align}
Equation~\Cref{eqn:beta_bound_2} follows from the definition
of $\cD(\hat{\bm{\ell}}^\lambda(\mat{A}) , d_\eff')$;
Equation~\Cref{eqn:beta_bound_3} follows since $\hat{\bm{\ell}}^\lambda(\mat{A})$
is a $\beta$-overestimate for the $\lambda$-ridge leverage score distribution
of $\mat{A}$;
and
Equation~\Cref{eqn:beta_bound_4} follows from the alternate characterization
of $\lambda$-ridge leverage scores in \Cref{lem:augmented_design_matrix}.

Further, each row index $j \in \{n+1,n+2,\dots,n+d\}$ is sampled with
probability
\begin{align}
  \label{eqn:beta_bound_5}
  \Pr(Y = j)
  &= 
  \frac{ \min\set*{1,d - d_\eff'} }{\norm{\hat{\bm{\ell}}^\lambda (\mat{A})}_{1} + d\min\set*{1,d - d_\eff'}} \\
  &=
  \frac{d}{\norm{\hat{\bm{\ell}}^\lambda (\mat{A})}_{1} + d\min\set{1,d - d_\eff'}}
  \cdot
  \frac{ \min\set*{1,d - d_\eff'} }{d} \notag \\
  \label{eqn:beta_bound_6}
  &\ge
  \parens*{\frac{ \norm{\hat{\bm{\ell}}^\lambda (\mat{A})}_{1} }{d} + \min\set{1,d - d_\eff'}}^{-1} \Pr(X = j).
\end{align}
Equation~\Cref{eqn:beta_bound_5} follows from the definition
of $\cD(\hat{\bm{\ell}}^\lambda(\mat{A}) , d_\eff')$, and
Equation~\Cref{eqn:beta_bound_6} uses the upper bound for
$\Pr(X=j)$ in Equation~\Cref{eqn:beta_bound_1}.
Therefore, we have
\begin{equation}
  \beta' =
  \min\set*{
    \parens*{1 + \frac{d \min\set{1, d-d_\eff'}}{\norm{\hat{\bm{\ell}}^\lambda(\mat{A})}_{1}} }^{-1}
    \frac{\beta d}{d_\eff}
    ,
  \parens*{\frac{\norm{\hat{\bm{\ell}}^\lambda(\mat{A})}_{1}}{d}
  + \min\set{1, d-d_\eff'}}^{-1}
  },
\end{equation}
as desired.
\end{proof}

\BetaOverestimateCorollary*

\begin{proof}
The first part of the claim follows from the inequality
\[
  \frac{\ell_{i}\parens*{\mat{A}}}{\norm{ \bm{\ell}\parens*{\mat{A}} }_1}
  =
  \frac{\ell_{i}\parens*{\mat{A}}}{\rank(\mat{A})}
  \ge
  \frac{\ell_{i}^\lambda\parens*{\mat{A}}}{\rank(\mat{A})}
  =
  \frac{d_\eff}{\rank(\mat{A})}\cdot
  \frac{\ell_{i}^\lambda\parens*{\mat{A}}}{d_\eff}.
\]
The second part is then a consequence of \Cref{lemma:beta_overestimate}.
Letting $\hat{\bm{\ell}}^\lambda(\mat{A}) = \bm{\ell}(\mat{A})$
and $\beta = d_\eff / \rank(\mat{A})$, we have
\begin{align*}
  \beta'
  &=  
    \parens*{
      \frac{d_\eff}{d} \cdot \frac{\rank(\mat{A})}{d_\eff}
      + \min\set*{1, d - d_\eff}
    }^{-1}
  =  
    \parens*{
      \frac{\rank(\mat{A})}{d}
      + \min\set*{1, d - d_\eff}
    }^{-1},
\end{align*}
which completes the proof.
\end{proof}

Now we prove our main theorem in \Cref{sec:approximate_ridge_regression},
which (1) shows that the solution vector to the sketched regression problem
is a good approximation to the initial ridge regression problem, and (2) 
quantifies how many row samples are needed for a $(1+\varepsilon)$-approximation
as a function of the input overestimate factor $\beta$.

\ApproximateRidgeRegression*

\begin{proof}
It suffices to find a $(1+\varepsilon)$-approximate solution to the
equivalent least squares problem
\[
  \argmin_{\mat{x} \in \R^{d}} \norm*{\overline{\mat{A}} \mat{x} - \overline{\mat{b}}}_{2}^2.
\]
We know that the augmented distribution $\mathcal{D}(\hat{\bm{\ell}}^\lambda(\mat{A}), d_\eff')$
is a $\beta'$-overestimate 
for the leverage score distribution of $\overline{\mat{A}}$
by \Cref{lemma:beta_overestimate}.
This enables us to use standard leverage score-based sampling arguments to
show that the sketching matrix $\mat{S} \in \R^{s \times n}$ satisfies
the two desired structural conditions in \Cref{subsec:fast_least_squares}.
Once we prove that Equations~\Cref{eqn:structural_condition_1} and
\Cref{eqn:structural_condition_2} both hold with probability at least $1-\delta$,
we apply \Cref{lem:approx_least_squares} to conclude that the
solution vector to the sketched least squares problem
\[
  \mat{\tilde x}_{\opt} = \argmin_{\mat{x} \in \R^d}
  \norm*{
    \mat{S}\overline{\mat{A}} \mat{x} - \mat{S}\overline{\mat{b}}
  }_{2}^2
\]
is a $(1+\varepsilon)$-approximation to the original ridge regression
problem.
Therefore, it remains to show that the structural conditions hold with
sufficient probability.

\textbf{Structural Condition 1.}
We show that Equation~\Cref{eqn:structural_condition_1} holds
for $\overline{\mat{A}}$
with probability at least $1 - \delta/2$ by using \Cref{thm:structural_tool_1}.
First, observe that the sampling matrix $\mat{\Omega}$
and rescaling matrix $\mat{D}$ in this theorem statement are a decomposition
of the sketching matrix in \Cref{alg:approximate_ridge_regression}
(i.e., $\mat{D}^\intercal \mat{\Omega}^\intercal = \mat{S}$).
Second, observe that it suffices to set $\varepsilon = 1 - 1/\sqrt{2}$
in the statement of \Cref{thm:structural_tool_1} in order to satisfy
\Cref{eqn:structural_condition_1}.

Let the compact SVD of $\overline{\mat{A}}$ be
$\overline{\mat{A}} = \mat{U}_{\overline{\mat{A}}} \mat{\Sigma}_{\overline{\mat{A}}} \mat{V}^\intercal_{\overline{\mat{A}}}$.
Since $\mathcal{D}(\hat{\bm{\ell}}^\lambda(\mat{A}), d_\eff')$
is a $\beta'$-overestimate 
for the leverage score distribution of $\overline{\mat{A}}$, it follows
from \Cref{thm:structural_tool_1} that with probability at least $1 - \delta/2$,
\[
  \sigma_{\min}^2\parens*{\mat{S} \mat{U}_{\overline{\mat{A}}}} \ge 1 / \sqrt{2},
\]
as long as the number of samples $s$ is at least
\begin{align*}
  s > \frac{144d\ln(4d/\delta)}{\beta'(1-1/\sqrt{2})^2}
    >
    \frac{1680 d \ln(4d/\delta)}{\beta'}.
\end{align*}
The algorithm rescales $\hat{\bm{\ell}}^\lambda(\mat{A})$ on Line~2
so that $\norm{\hat{\bm{\ell}}^\lambda(\mat{A})}_1 = d$.
Therefore, it follows from
\Cref{lemma:beta_overestimate} and the fact that $d_\eff \le d$
that
\begin{align}
  \beta' &= 
  \min\set*{
    \parens*{1 + \frac{d \min\set{1, d-d_\eff'}}{\norm{\hat{\bm{\ell}}^\lambda(\mat{A})}_{1}} }^{-1}
    \frac{\beta d}{d_\eff}
    ,
  \parens*{\frac{\norm{\hat{\bm{\ell}}^\lambda(\mat{A})}_{1}}{d}
  + \min\set{1, d-d_\eff'}}^{-1}
  } \notag \\
  &\ge
  \label{eqn:beta_prime_bound_for_algorithm}
  \min\set*{\beta, 1} / \parens*{1 + \min\set*{1, d - d'_\eff}}.
\end{align}
Thus, the number of samples is
sufficient by our choice of $\beta'$ and $s$ 
on Lines~2--3 of \Cref{alg:approximate_ridge_regression},
so
Equation~\Cref{eqn:structural_condition_1} holds
for $\overline{\mat{A}}$ with probability at least $1 - \delta/2$.

\textbf{Structural Condition 2.}
We show that 
Equation~\Cref{eqn:structural_condition_2} holds
for $\overline{\mat{A}}$
with probability at least $1 - \delta/2$ by
using \Cref{thm:structural_tool_2} and Markov's inequality.
First observe that
\[
  \mat{U}_{\overline{\mat{A}}}^\intercal \overline{\mat{b}}^\perp
  =
  \mat{U}_{\overline{\mat{A}}}^\intercal
  \parens*{
  \mat{U}_{\overline{\mat{A}}}^\perp
  {\mat{U}_{\overline{\mat{A}}}^\perp}^\intercal
  \overline{\mat{b}}
  }
  =
  \mat{0}_{\rank\parens*{\overline{\mat{A}}}},
\]
where $\overline{\mat{b}}^\perp$ is defined analogously to $\mat{b}^\perp$
in \Cref{subsec:fast_least_squares}.
Thus, the second structural condition can be seen as
bounding how closely the sampled product approximates the zero vector.
To do this, we can apply \Cref{thm:structural_tool_2}
since the leverage scores of a matrix $\mat{M}$ are the squared distances of the
row vectors of its compact left-singular matrix $\mat{U}_{\mat{M}}$.
Combining these facts, it follows that
\begin{align*}
  \E\bracks*{ \norm*{ \mat{U}_{\overline{\mat{A}}}^\intercal
      \mat{S}^\intercal \mat{S} \overline{\mat{b}}^\intercal }_{2}^2 }
  &=
  \E\bracks*{ \norm*{
    \mat{0}_{\rank\parens*{\overline{\mat{A}}}} - 
    \mat{U}_{\overline{\mat{A}}}^\intercal
      \mat{S}^\intercal \mat{S} \overline{\mat{b}}^\intercal }_{2}^2 } \\
  &=
  \E\bracks*{ \norm*{
    \mat{U}_{\overline{\mat{A}}}^\intercal \overline{\mat{b}}^\perp - 
    \mat{U}_{\overline{\mat{A}}}^\intercal
      \mat{S}^\intercal \mat{S} \overline{\mat{b}}^\intercal }_{2}^2 } \\
  &\le
  \frac{1}{\beta' s} \norm*{\mat{U}_{\overline{\mat{A}}}^\intercal}_{\frobenius}^2 \norm*{\overline{\mat{b}}^\perp }_{2}^2 \\
  &=
  \frac{d}{\beta' s} \norm*{\overline{\mat{b}}^\perp }_{2}^2.
\end{align*}
Since the norm of the residual vector equals the norm of $\overline{\mat{b}}^\perp$,
applying Markov's inequality gives us
\begin{align}
\label{eqn:structural_2_sample_bound}
  \Pr\parens*{
     \norm*{ \mat{U}_{\overline{\mat{A}}}^\intercal
      \mat{S}^\intercal \mat{S} \overline{\mat{b}}^\intercal }_{2}^2
      \ge
      \frac{\varepsilon  \norm{\overline{\mat{b}}^\perp}_{2}^2}{2}
  }
  \le
  \frac{2 \E\bracks*{ \norm{ \mat{U}_{\overline{\mat{A}}}^\intercal
      \mat{S}^\intercal \mat{S} \overline{\mat{b}}^\intercal }_{2}^2 } }{ \varepsilon  \norm{\overline{\mat{b}}^\perp}_{2}^2 }
  \le
  \frac{2d}{\varepsilon \beta' s}.
\end{align}
Bounding the rightmost expression of \Cref{eqn:structural_2_sample_bound}
by the failure probability $\delta / 2$
implies that the number of samples must be at least
\[
  s \ge \frac{4d}{\beta' \delta \varepsilon}.
\]
The number of samples $s$ in \ApproxRidgeRegression~is 
large enough by \Cref{eqn:beta_prime_bound_for_algorithm} and 
Line~3 of \Cref{alg:approximate_ridge_regression}.
Therefore, 
Equation~\Cref{eqn:structural_condition_2} holds
for $\overline{\mat{A}}$ with probability at least $1 - \delta/2$.

Taking a union bound over the failure probabilities for the two structural
conditions ensures that both are satisfied with probability at least
$1 - (\delta / 2 + \delta/2) = 1 - \delta$.
Furthermore, since the running times of the two subroutines are parameterized
by $t$ and $T(n',d')$, the result immediately follows.
\end{proof}

\section{Missing Analysis from \Cref{sec:algorithm}}
\label{app:algorithm}

Now we prove several technical results about $\lambda$-ridge leverage scores
and their upper bounds as rows are removed from the design matrix.
First, we present an explicit formula for the $\lambda$-ridge leverage scores
of a Kronecker matrix in terms of the singular value decompositions of its
factor matrices.

\KroneckerCrossLeverageScores*

\begin{proof}
For notational brevity, we prove the
claim for $\mat{K} = \mat{A} \ktimes \mat{B} \ktimes \mat{C}$.
The order-$N$ version follows by the same argument.

First, the mixed property property of Kronecker products implies that
\begin{align*}
    \mat{K}^\intercal \mat{K}
    &=
    \parens*{\mat{A}^\intercal \mat{A}} \ktimes \parens*{\mat{B}^\intercal \mat{B}} \ktimes \parens*{\mat{C}^\intercal \mat{C}}.
\end{align*}
Let $\mat{A} = \mat{U}_{\mat{A}} \mat{\Sigma}_{\mat{A}} \mat{V}_{\mat{A}}^\intercal$
denote the SVD of $\mat{A}$
such that $\mat{U}_{\mat{A}} \in \R^{I_1 \times I_1}$ and $\mat{V}_{\mat{A}} \in \R^{R_1 \times R_1}$.
It follows from the orthogonality of $\mat{U}_{\mat{A}}$ that
\[
    \mat{A}^\intercal \mat{A}
    = \mat{V}_{\mat{A}} \mat{\Sigma}_{\mat{A}}^2 \mat{V}_{\mat{A}}^\intercal,
\]
where $\mat{\Sigma}_{\mat{A}}^2$ denotes $\mat{\Sigma}_{\mat{A}}^\intercal \mat{\Sigma}_{\mat{A}}$.
Similarly, let
$\mat{B} = \mat{U}_{\mat{B}} \mat{\Sigma}_{\mat{B}} \mat{V}_{\mat{B}}^\intercal$
and
$\mat{C} = \mat{U}_{\mat{C}} \mat{\Sigma}_{\mat{C}} \mat{V}_{\mat{C}}^\intercal$.
It follows from the mixed-product property that 
\begin{align*}
    \mat{K}^\intercal \mat{K} + \lambda \mat{I}
    &=
    \parens*{\mat{V}_{\mat{A}} \mat{\Sigma}^2_{\mat{A}} \mat{V}_{\mat{A}}^\intercal}
    \ktimes
    \parens*{\mat{V}_{\mat{B}} \mat{\Sigma}^2_{\mat{B}} \mat{V}_{\mat{B}}^\intercal}
    \ktimes
    \parens*{\mat{V}_{\mat{C}} \mat{\Sigma}^2_{\mat{C}} \mat{V}_{\mat{C}}^\intercal}
    + \lambda \mat{I} \\
    &=
    \parens*{\mat{V}_{\mat{A}} \ktimes \mat{V}_{\mat{B}} \ktimes \mat{V}_{\mat{C}}}
    \parens*{\mat{\Sigma}^2_{\mat{A}} \ktimes \mat{\Sigma}^2_{\mat{B}} \ktimes \mat{\Sigma}^2_{\mat{C}}}
    \parens*{\mat{V}_{\mat{A}}^\intercal \ktimes \mat{V}_{\mat{B}}^\intercal \ktimes \mat{V}_{\mat{C}}^\intercal}
    + \lambda \mat{I} \\
    &=
    \parens*{\mat{V}_{\mat{A}} \ktimes \mat{V}_{\mat{B}} \ktimes \mat{V}_{\mat{C}}}
    \parens*{\mat{\Sigma}^2_{\mat{A}} \ktimes \mat{\Sigma}^2_{\mat{B}} \ktimes \mat{\Sigma}^2_{\mat{C}} + \lambda \mat{I}}
    \parens*{\mat{V}_{\mat{A}}^\intercal \ktimes \mat{V}_{\mat{B}}^\intercal \ktimes \mat{V}_{\mat{C}}^\intercal}.
\end{align*}
Since $\parens{\mat{X} \mat{Y}}^+ = \mat{Y}^+ \mat{X}^+$ if $\mat{X}$ or $\mat{Y}$ is orthogonal,
we have
\begin{align*}
    \parens*{\mat{K}^\intercal \mat{K} + \lambda \mat{I}}^{+}
    &=
    \parens*{\parens*{\mat{V}_{\mat{A}} \ktimes \mat{V}_{\mat{B}} \ktimes \mat{V}_{\mat{C}}}
    \parens*{\mat{\Sigma}^2_{\mat{A}} \ktimes \mat{\Sigma}^2_{\mat{B}} \ktimes \mat{\Sigma}^2_{\mat{C}} + \lambda \mat{I}}
    \parens*{\mat{V}_{\mat{A}}^\intercal \ktimes \mat{V}_{\mat{B}}^\intercal \ktimes \mat{V}_{\mat{C}}^\intercal}}^{+} \\
    &=
    \parens*{\mat{V}_{\mat{A}}^\intercal \ktimes \mat{V}_{\mat{B}}^\intercal \ktimes \mat{V}_{\mat{C}}^\intercal}^{+}
    \parens*{\mat{\Sigma}^2_{\mat{A}} \ktimes \mat{\Sigma}^2_{\mat{B}} \ktimes \mat{\Sigma}^2_{\mat{C}} + \lambda \mat{I}}^{+}
    \parens*{\mat{V}_{\mat{A}} \ktimes \mat{V}_{\mat{B}} \ktimes \mat{V}_{\mat{C}}}^{+} \\
    &=
    \parens*{\mat{V}_{\mat{A}} \ktimes \mat{V}_{\mat{B}} \ktimes \mat{V}_{\mat{C}}}
    \parens*{\mat{\Sigma}^2_{\mat{A}} \ktimes \mat{\Sigma}^2_{\mat{B}} \ktimes \mat{\Sigma}^2_{\mat{C}} + \lambda \mat{I}}^{+}
    \parens*{\mat{V}_{\mat{A}}^\intercal \ktimes \mat{V}_{\mat{B}}^\intercal \ktimes \mat{V}_{\mat{C}}^\intercal}.
\end{align*}
Next, observe that
\begin{align*}
    \mat{K}
    &=
    \parens*{\mat{U}_{\mat{A}} \mat{\Sigma}_{\mat{A}} \mat{V}_{\mat{A}}^\intercal}
    \ktimes
    \parens*{\mat{U}_{\mat{B}} \mat{\Sigma}_{\mat{B}} \mat{V}_{\mat{B}}^\intercal}
    \ktimes
    \parens*{\mat{U}_{\mat{C}} \mat{\Sigma}_{\mat{C}} \mat{V}_{\mat{C}}^\intercal} \\
    &=
    \parens*{
    \mat{U}_{\mat{A}} \ktimes \mat{U}_{\mat{B}} \ktimes \mat{U}_{\mat{C}}
    }
    \parens*{
    \mat{\Sigma}_{\mat{A}} \ktimes \mat{\Sigma}_{\mat{B}} \ktimes \mat{\Sigma}_{\mat{C}}
    }
    \parens*{
    \mat{V}_{\mat{A}}^\intercal \ktimes \mat{V}_{\mat{B}}^\intercal \ktimes \mat{V}_{\mat{C}}^\intercal
    }.
\end{align*}
Putting everything together, the $\lambda$-ridge cross leverage scores can be expressed as
\begin{align}
\label{eqn:leverage_score_matrix_eigendecomposition}
    \mat{K}\parens*{\mat{K}^\intercal \mat{K} + \lambda \mat{I}}^{+} \mat{K}^\intercal
    &=
    \parens*{
        \mat{U}_{\mat{A}} \ktimes \mat{U}_{\mat{B}} \ktimes \mat{U}_{\mat{C}}
    }
    \mat{\Lambda}
    \parens*{
        \mat{U}_{\mat{A}} \ktimes \mat{U}_{\mat{B}} \ktimes \mat{U}_{\mat{C}}
    }^\intercal,
\end{align}
where
\[
    \mat{\Lambda} =
    \parens*{\mat{\Sigma}_{\mat{A}} \ktimes \mat{\Sigma}_{\mat{B}} \ktimes \mat{\Sigma}_{\mat{C}}}
    \parens*{\mat{\Sigma}^2_{\mat{A}} \ktimes \mat{\Sigma}^2_{\mat{B}} \ktimes \mat{\Sigma}^2_{\mat{C}} + \lambda \mat{I}}^{+}
    \parens*{\mat{\Sigma}_{\mat{A}} \ktimes \mat{\Sigma}_{\mat{B}} \ktimes \mat{\Sigma}_{\mat{C}}}.
\]
Equation~\Cref{eqn:leverage_score_matrix_eigendecomposition} is the
eigendecomposition of
$\mat{K}\parens*{\mat{K}^\intercal \mat{K} + \lambda \mat{I}}^{+} \mat{K}^\intercal$.
In particular, $\mat{\Lambda} \in \R^{I_1 I_2 I_3 \times I_1 I_2 I_3}$
is a diagonal matrix of eigenvalues, where
the $(i_1,i_2,i_3)$-th eigenvalue is
\begin{align}
\label{eqn:leverage_score_matrix_eigenvalues}
  \lambda_{(i_1,i_2,i_3)} =
    \frac{\sigma_{i_1}^2\parens*{\mat{A}} \sigma_{i_2}^2\parens*{\mat{B}} \sigma_{i_3}^2\parens*{\mat{C}}}{\sigma_{i_1}^2\parens*{\mat{A}} \sigma_{i_2}^2\parens*{\mat{B}} \sigma_{i_3}^2\parens*{\mat{C}} + \lambda}.
\end{align}
The claim about the values of
$\ell^\lambda_{(i_1,i_2,i_3),(j_1,j_2,j_3)}\parens{\mat{K}}$
then follows from the definition of cross $\lambda$-ridge leverage scores
in Equation~\Cref{eqn:ridge_leverage_score_svd_def}.
Furthermore, the statistical leverage score property holds because
setting $\lambda = 0$ gives an expression that is the product of the
leverage scores of the factor matrices.
\end{proof}

Next we analyze the time and space complexity of the fast core tensor update
via our approximate ridge regression algorithm using leverage scores as
an overestimate for the $\lambda$-ridge leverage scores.

\FastCoreTensorUpdateRunningTime*

\begin{proof}
The leverage scores of the factor matrix $\mat{A}^{(n)}$ can be computed in
$O(I_n R_n^2)$ time using
the pseudoinverse expression in Equation~\Cref{eqn:ridge_leverage_score_def}.
Next, compute cumulative density functions for each of the factor distributions
in $O(R_1 + R_2 + \cdots + R_N)$ time. This enables us to sample a row index of
from the leverage score distribution of $\mat{K}$ in time
$O(\log(R_1) + \log(R_2) + \cdots + \log(R_N))$ using $N$ binary searches
and the factorization property in \Cref{lemma:kronecker_cross_leverage_scores}.
The number of samples in the approximate ridge regression subroutine is
$s = O(R \max\{\log(R/\delta), 1/(\delta \varepsilon)\})$
since $\beta' \ge 1/2$ if leverage scores are used as the $\lambda$-ridge
leverage score overestimate by \Cref{cor:leverage_score_overestimate}.
Therefore, we can sample all of the row indices of
the ridge-augmented version of matrix $\mat{K}$ and
construct the sketch matrix $\mat{S}$ in
time $O(R \log(R) \max\{\log(R/\delta), 1/(\delta \varepsilon)\})$.

We solve the least squares problem with design matrix
$\mat{A} = \mat{S}\overline{\mat{K}}$ and response vector
$\mat{b} = \mat{S}\overline{\vectorize(\mat{X})}$
by computing $\mat{A}^\intercal \mat{A}$ and $\mat{A}^\intercal \mat{b}$
and using the normal equations.
The matrix $\mat{A}^\intercal \mat{A}$ is the outer product of
$s = O(R \max\{\log(R/\delta), 1/(\delta \varepsilon)\})$ row vectors of length $R$,
and can therefore be constructed in
$O(s R^2) = O(R^3 \max\{\log(R/\delta), 1/(\delta \varepsilon)\})$
time using the data structure to implicitly query entries of $\mat{K}$
from the factor matrices. 
Note that we first compute a row of $\mat{K}$ using this data structure, and
then we compute the outer product in $O(R^2)$ time.
We can compute $(\mat{A}^\intercal \mat{A})^+$ in $O(R^{\omega})$
time via~\cite{alman2021refined},
where $\omega < 2.737$, since $\mat{A}^\intercal \mat{A}$ is an $R \times R$ matrix.

We construct $\mat{A}^\intercal \mat{b}$ similarly in
$O(R \cdot R \max\{\log(R/\delta), 1/(\delta \varepsilon)\})$ time.
Putting everything together and using the normal equations,
we can compute
$\mat{\tilde x}_{\opt} = (\mat{A}^\intercal \mat{A})^+ \mat{A}^\intercal \mat{b}$
in $O(R^2)$ time.
Therefore, the total running time of \Cref{alg:fast_core_tensor_update} is
$O(R^\omega \max\{\log(R/\delta), 1/(\delta \varepsilon)\})$.

The approximation guarantee follows directly from
\Cref{cor:leverage_score_overestimate} and
\Cref{thm:approximate_ridge_regression}
since the augmented leverage scores of $\mat{K}$
are a $(1/2)$-overestimate of the true $\lambda$-ridge leverage scores.
We achieve the claimed space complexity by
using implicit representations of
$\mat{K}$ and its leverage scores.
\end{proof}

Now we prove two elucidatory upper bounds for $\lambda$-ridge leverage
scores as rows are removed from the design matrix.
The first upper bound crucially relies on the regularization strength
$\lambda$ being positive, but an analogous statement holds if we use the
pseudoinverse all the way throughout the proof.
This pseudoinverse-based result is useful for accurately bounding
statistical leverage scores (i.e., $\lambda = 0$)
when removing a set of rows lowers the rank of the matrix.
The second upper bound holds for any set of removed rows, and thus is not
particularly strong in general. That said, if the factor matrices are
well-structured (e.g., columnwise orthogonality constraints),
then it provides nontrivial guarantees.

\RidgeScoreUpperBoundAfterRemovingRows*

\begin{proof}
Let $\mat{W}$ be the $0$-$1$ diagonal matrix such that
$w_{ii}=1$ if $i\in S$ and $w_{ii}=0$ otherwise. Then
\[
    \tilde{\mat{A}}^\intercal \tilde{\mat{A}}
    =
    \mat{A}^\intercal \mat{A} - ((\mat{I}-\mat{W})\mat{A})^\intercal ((\mat{I}-\mat{W})\mat{A}).
\]
Let $\mat{B} = (\mat{I}-\mat{W})\mat{A}$.
For any $i \in S$, it follows that
\begin{align*}
    \ell^{\lambda}_i \parens*{\mat{\tilde A}}
    &=
    \mat{a}_{i:} \parens*{\mat{\tilde A}^\intercal \mat{\tilde A} + \lambda \mat{I}}^{+} \mat{a}_{i:}^\intercal \\
    &=
    \mat{a}_{i:} \parens*{\mat{A}^\intercal \mat{A} - \mat{B}^\intercal \mat{B} + \lambda \mat{I}}^{+} \mat{a}_{i:}^\intercal.
\end{align*}
The matrix
$\mat{A}^\intercal \mat{A} - \mat{B}^\intercal \mat{B} + \lambda \mat{I}$ is positive definite
because $\lambda > 0$, and hence invertible.
Therefore, by the Woodbury matrix identity,
we have
\begin{align*}
    \parens*{\mat{A}^\intercal \mat{A} - \mat{B}^\intercal \mat{B} + \lambda \mat{I}}^{-1}
    &=
    \parens*{\mat{A}^\intercal \mat{A} + \lambda \mat{I}}^{-1} \\
    &\hspace{0.40cm}+ \parens*{\mat{A}^\intercal \mat{A} + \lambda \mat{I}}^{-1}
        \mat{B}^\intercal
        \parens*{\mat{I} - \mat{B} \parens*{\mat{A}^\intercal \mat{A} 
                 + \lambda\mat{I}}^{-1} \mat{B}^{\intercal}}^{-1}
                 \mat{B} \parens*{\mat{A}^\intercal \mat{A} + \lambda \mat{I}}^{-1}.
\end{align*}
For each $i \in S$, let $\mat{v}_{i:}\in \R^{1 \times n}$ be the row vector whose
$j$-th entry is equal to $\ell^{\lambda}_{ij}\parens{\mat{A}}$
if $j \in \overline{S}$ and zero otherwise. Then
\[
  \mat{v}_{i:}
  = \mat{a}_{i:} \parens*{\mat{A}^\intercal \mat{A} + \lambda \mat{I}}^{-1} \mat{B}^{\intercal}.
\]
Hence, we have
\begin{align*}
    &
    \mat{a}_{i:}\parens*{\mat{A}^\intercal \mat{A} - \mat{B}^\intercal \mat{B} + \lambda\mat{I}}^{-1} \mat{a}_{i:}^{\intercal} \\
    &\hspace{0.15cm}=
    \mat{a}_{i:} \bracks*{
      \parens*{\mat{A}^\intercal \mat{A} + \lambda \mat{I}}^{-1}
      + 
      \parens*{\mat{A}^\intercal \mat{A} + \lambda \mat{I}}^{-1}
    \mat{B}^\intercal
    \parens*{\mat{I} - \mat{B} \parens*{\mat{A}^\intercal \mat{A} + \lambda\mat{I}}^{-1} \mat{B}^{\intercal}}^{-1} \mat{B} \parens*{\mat{A}^\intercal \mat{A} + \lambda \mat{I}}^{-1}} \mat{a}_{i:}^{\intercal} \\
    &\hspace{0.15cm}=
    \ell^{\lambda}_i\parens*{\mat{A}} + \mat{v}_{i:}
        \parens*{\mat{I} - \mat{B} \parens*{\mat{A}^\intercal \mat{A} + \lambda\mat{I}}^{-1} \mat{B}^{\intercal}}^{-1} \mat{v}_{i:}^\intercal.
\end{align*}
In terms of principal submatrices, this is equivalent to
\begin{align}
  \label{eqn:ridge_score_update_formula}
  \ell_{i}^\lambda \parens*{\mat{\tilde A}}
  =
  \ell_{i}^\lambda \parens*{\mat{A}}
  +
  \bm{\ell}^\lambda_{\overline{S}}\parens*{\mat{A}}^\intercal
  \parens*{\mat{I} - \mat{L}_{\overline{S},\overline{S}}}^{-1}
  \bm{\ell}^\lambda_{\overline{S}}\parens*{\mat{A}}.
\end{align}
Therefore, it follows from the Rayleigh--Ritz theorem that
\begin{align*}
  \ell_{i}^\lambda \parens*{\mat{\tilde A}}
  \le
  \lambda_{\max}\parens*{\parens*{\mat{I} - \mat{L}_{\overline{S},\overline{S}}}^{-1}}
  \norm*{\bm{\ell}^\lambda_{\overline{S}}\parens*{\mat{A}}}_{2}^2
  =
  \lambda_{\max}\parens*{\parens*{\mat{I} - \mat{L}_{\overline{S},\overline{S}}}^{-1}}
  \sum_{j \in \overline{S}} \ell_{ij}^\lambda\parens*{\mat{A}}^2.
\end{align*}
We know from \Cref{eqn:leverage_score_matrix_eigenvalues}
that all of the eigenvalues of $\mat{L}$ are nonnegative and
strictly less than one since
\[
  \lambda_{i}\parens*{\mat{L}} = \frac{\sigma_{i}^2\parens*{\mat{A}}}{\sigma_{i}^2\parens*{\mat{A}} + \lambda}.
\]
Therefore,
since $\mat{L}$ is symmetric and
$\mat{L}_{\overline{S},\overline{S}}$ is a principal submatrix of $\mat{L}$,
Cauchy's interlacing theorem implies that
\begin{equation}
\label{eqn:interlacing}
  0 \le \lambda_{i}\parens*{\mat{L}_{\overline{S},\overline{S}}}
  \le
  \lambda_{i}\parens*{\mat{L}} < 1.
\end{equation}
Finally, it follows from a standard eigenvalue argument that
\begin{align*}
  \lambda_{\max}\parens*{ \parens*{\mat{I} - \mat{L}_{\overline{S},\overline{S}}}^{-1} }
  =
  \frac{1}{1 - \lambda_{\max}\parens*{\mat{L}_{\overline{S},\overline{S}}}},
\end{align*}
which completes the proof.
\end{proof}


\begin{restatable}{corollary}{MaxEigenvalueUpperBound}
If
$\mat{K} = \mat{A}^{(1)} \otimes \mat{A}^{(2)} \otimes \dots \otimes \mat{A}^{(N)}$
where $\mat{A}^{(n)} \in \R^{I_n \times R_n}$,
$\mat{L} = \mat{K}\parens{\mat{K}^\intercal \mat{K} + \lambda\mat{I}}^+\mat{K}^\intercal$
is the matrix of cross $\lambda$-ridge leverage scores of $\mat{K}$,
then for any $\lambda > 0$ and $S \subseteq [n]$, we have
\[
  \frac{1}{1 - \lambda_{\max}\parens*{\mat{L}_{\overline{S},\overline{S}}}}
  \le
  1 + \frac{ \prod_{n=1}^N
    \norm*{\mat{A}^{(n)}}_{2}^2}{\lambda}.
\]
\end{restatable}

\begin{proof}
The claim is a direct consequence of the formula for the
eigenvalues of $\mat{L}$ in Equation~\Cref{eqn:leverage_score_matrix_eigenvalues}
and
the application of
Cauchy's interlacing theorem in Equation~\Cref{eqn:interlacing}.
Specifically, we know that
\begin{align*}
  \lambda_{\max}\parens*{\mat{L}_{\overline{S},\overline{S}}}
  \le
  \lambda_{\max}\parens*{\mat{L}}
  =
  \frac{\prod_{n=1}^N \sigma^2_{\max}\parens*{\mat{A}^{(n)}}}{\prod_{n=1}^N \sigma^2_{\max}\parens*{\mat{A}^{(n)}} + \lambda}
  =
  \frac{\prod_{n=1}^N \norm*{\mat{A}^{(n)}}_{2}^2 }{\prod_{n=1}^N \norm*{\mat{A}^{(n)}}_{2}^2 + \lambda}.
\end{align*}
Therefore, it follows that
\[
 \frac{1}{1 - \lambda_{\max}\parens*{\mat{L}_{\overline{S},\overline{S}}}}
  \le 
 \frac{\prod_{n=1}^N \norm*{\mat{A}^{(n)}}_{2}^2 + \lambda}{\lambda}
 =
  1 + \frac{ \prod_{n=1}^N \norm*{\mat{A}^{(n)}}_{2}^2 }{\lambda}.
  \qedhere
\]
\end{proof}

Last, we prove a property of ridge leverage scores
that is closely related to the fact that hat matrices in linear regression
are idempotent.
This inequality can be combined with the upper bounds above to give
multiplicative bounds for $\lambda$-ridge leverage scores of a matrix after
some of its rows are removed,
and the coefficient in these multiplicative bounds
is approximately equal to $\beta^{-1}$ in the $\beta$-overestimate.

\SumOfSquaredCrossScores*

\begin{proof}
First, observe that
\begin{align*}
  \sum_{j=1}^n \ell_{ij}^\lambda \parens*{\mat{A}}^2
  &=
  \sum_{j=1}^n \parens*{
    \sum_{k=1}^n \frac{\sigma_{k}^2\parens*{\mat{A}}}{\sigma_{k}^2\parens*{\mat{A}} + \lambda}
    u_{ik} u_{jk}
  }^2 \\
  &=
  \sum_{j=1}^n
  \sum_{k_1=1}^{n}
  \sum_{k_2=1}^{n}
    \parens*{\frac{\sigma_{k_1}^2\parens*{\mat{A}}}{\sigma_{k_1}^2\parens*{\mat{A}} + \lambda}}
    \parens*{\frac{\sigma_{k_2}^2\parens*{\mat{A}}}{\sigma_{k_2}^2\parens*{\mat{A}} + \lambda}}
    u_{i k_1} u_{j k_1} u_{i k_2} u_{j k_2} \\
  &=
  \sum_{k_1=1}^{n}
  \sum_{k_2=1}^{n}
    \parens*{\frac{\sigma_{k_1}^2\parens*{\mat{A}}}{\sigma_{k_1}^2\parens*{\mat{A}} + \lambda}}
    \parens*{\frac{\sigma_{k_2}^2\parens*{\mat{A}}}{\sigma_{k_2}^2\parens*{\mat{A}} + \lambda}}
    u_{i k_1} u_{i k_2}
  \sum_{j=1}^n
    u_{j k_1} u_{j k_2}.
\end{align*}
Since $\mat{U}$ is an orthogonal matrix, we have
\[
  \sum_{j=1}^n u_{j k_1} u_{j k_2}
  =
  \begin{cases}
    1 & \text{if $k_1 = k_2$}, \\
    0 & \text{if $k_2 \ne k_2$}.
  \end{cases}
\]
Therefore, it follows that
\begin{align*}
  \label{eqn:cross_scores_squared}
  \sum_{j=1}^n \ell_{ij}^\lambda \parens*{\mat{A}}^2
  &=
  \sum_{k_1=1}^{n}
    \parens*{\frac{\sigma_{k}^2\parens*{\mat{A}}}{\sigma_{k}^2\parens*{\mat{A}} + \lambda}}^2
    u_{i k}^2
  \le
  \sum_{k_1=1}^{n}
    \frac{\sigma_{k}^2\parens*{\mat{A}}}{\sigma_{k}^2\parens*{\mat{A}} + \lambda}
    u_{i k}^2
  =
  \ell_{i}^\lambda \parens*{\mat{A}}.
\end{align*}
Moreover, the inequality above implies that
the two expressions are equal if and only if $\lambda = 0$.
\end{proof}

\end{document}